\documentclass[11pt,letterpaper]{article}
\usepackage{amsmath,amsfonts,amsthm,amssymb}
\usepackage{setspace}
\usepackage{fancyhdr}
\usepackage{lastpage}
\usepackage{extramarks}
\usepackage{xspace}
\usepackage{chngpage}
\usepackage{comment}

\usepackage[ruled,linesnumbered]{algorithm2e} 

\usepackage{soul,color}
\usepackage{graphicx,float,wrapfig}
\usepackage[font=small,labelfont=bf]{caption}
\usepackage[margin=1in]{geometry}
\usepackage{enumitem}
\usepackage{thmtools,thm-restate}
\usepackage{bbm}

\setlist{nosep}

\allowdisplaybreaks

\usepackage[dvipsnames]{xcolor}
\usepackage[linktocpage=true,breaklinks,colorlinks,citecolor=blue,linkcolor=BrickRed]{hyperref}

\usepackage{wrapfig}
  
\title{\textmd{\bf Online Advertising with Spatial Interactions  }}

\date{\today}
\author{
Gagan Aggarwal\thanks{Google Research (\texttt{gagana@google.com}). }
\and
Yifan Wang\thanks{School of Computer Science, Georgia Tech (\texttt{ywang3782@gatech.edu}). This work was done while the author was visiting Google as a Student Researcher.}
\and
Mingfei Zhao\thanks{Google Research (\texttt{mingfei@google.com}).}
}

\usepackage[usenames,dvipsnames]{xcolor}
\usepackage[linktocpage=true,breaklinks,colorlinks,citecolor=blue,linkcolor=BrickRed]{hyperref}
\usepackage{cleveref}

\newlist{constraints}{enumerate}{1}
\setlist[constraints]{%
  label=\textbullet,             
  leftmargin=2em,
  labelsep=0.5em,
  align=left,
  ref=\Roman*                    
}

\crefname{constraintsi}{Constraint}{Constraints}
\Crefname{constraintsi}{Constraint}{Constraints}

\crefformat{constraintsi}{Constraint~(#2#1#3)}
\Crefformat{constraintsi}{Constraint~(#2#1#3)}

\DeclareUnicodeCharacter{2217}{*}

\newtheorem{Theorem}{Theorem}[section]
\newtheorem{Lemma}[Theorem]{Lemma}

\newtheorem{Definition}[Theorem]{Definition}

\newtheorem{Claim}[Theorem]{Claim}
\newtheorem{Result}[Theorem]{Result}

\newcommand{\parta}{(\text{\uppercase\expandafter{\romannumeral1}})}
\newcommand{\partb}{(\text{\uppercase\expandafter{\romannumeral2}})}
\newcommand{\partc}{(\text{\uppercase\expandafter{\romannumeral3}})}

\newcommand{\alg}{\mathsf{Alg}}

\newcommand{\X}{\mathcal{X}}

\newcommand{\E}{\mathbb{E}}

\newcommand{\sw}{\mathsf{SW}}
\newcommand{\val}{\mathsf{Val}}
\renewcommand{\deg}{\mathsf{deg}}
\newcommand{\mul}{\text{mul}}

\newcommand{\OTild}{\widetilde{O}}
\newcommand{\one}{\mathbf{1}}%
\newcommand{\pr}{\mathbf{Pr}} 

\newcommand{\ignore}[1]{{}}
\newcommand{\R}{\mathbb{R}}

\newcommand{\D}{\mathcal{D}}

\newcommand{\poly}{\mathsf{poly}}

\newcommand{\IGNORE}[1]{}
\newcounter{note}[section]

\begin{document}


\maketitle \thispagestyle{empty}

\begin{abstract}

Online advertising platforms must decide how to allocate multiple ads across limited screen real estate, where each ad’s effectiveness depends not only on its own placement but also on nearby ads competing for user attention. Such spatial externalities — arising from proximity, clutter, or crowding — can significantly alter welfare and revenue outcomes, yet existing auction and allocation models typically treat ad slots as independent or ordered along a single dimension.

We introduce a new framework for spatial externalities in online advertising, in which the value of an ad depends on both its slot and the configuration of surrounding ads. We model ad slots as points in a metric space, and model an advertiser's value as a function of both their bid and a discount factor determined by the configuration of other displayed ads. Within this framework, we analyze two natural models. For the Nearest-Neighbor model, where the value suppression depends only on the closest neighboring ad, we present a polynomial-time algorithm that achieves a constant approximation for the general case. We show that the allocation rule is monotone and can be implemented as a truthful mechanism. For a structured setting of 2D Euclidean space, we provide a PTAS. In contrast, for the Product-Distance model, where interference is aggregated multiplicatively across all neighbors, we establish a strong (and nearly-tight) hardness of approximation --  no polynomial-time algorithm can achieve any polynomial-factor approximation unless P=NP, via a reduction from Max-Independent-Set.

Our results provide a foundation for reasoning about spatial externalities in ad allocation and for designing efficient, truthful mechanisms under such interactions.

\end{abstract}

\newpage

\section{Introduction}




Online advertising, a half-trillion-dollar industry \cite{EMK-2023}, represents a cornerstone of the modern digital economy. Ads are typically presented as a slate -- multiple creatives shown simultaneously in distinct locations on a page -- and chosen via real-time auctions. Platforms aim to maximize revenue or welfare, commonly achieved through optimizing expected yield, a metric combining the advertiser's bid with the predicted  Click-Through Rate (CTR) \cite{AGM-EC06,EOS-AER07,Varian-2007}.


A standard approach is to auction off every available ad slot as if its yield were independent of other ad slots -- evaluating an ad placed in the slot independently of ads placed in other slots. 
In practice, webpage real estate and user attention are scarce: when multiple ads are shown together, they compete for clicks, and an ad’s realized performance depends not only on its own placement but also on the number and placement of neighboring ads. The result is a web of negative 
externalities across the slate. This raises a natural question: \emph{how should an ad platform allocate ads across slots when the presence of one ad can have negative externalities on the effectiveness of other ads?}

Prior work has studied several important forms of such externalities. \cite{GM-WWW08} initiates the study of externalities in online advertising and introduces audience-side choice models that show that high-quality ads can steal attention from others.  Cascade and Markovian scan models capture how users traverse positions and how early interactions depress engagement with lower positions \cite{KM-WINE08, AFMP-WINE08}. \cite{HM-WINE14} develops a brand-effect model in position auctions, where the presence of strong brand ads generates asymmetric externalities on other advertisers. 
There is also a lot of empirical work that provides evidence of interference across ads. For example, \cite{CZTR-WSDM08} provides an experimental analysis of click position-bias models, showing that the performance of an ad depends critically on the surrounding slate. \cite{HRK-JAR09} shows how visual clutter on a webpage reduces advertising effectiveness, offering empirical evidence of inter-ad interference. While these frameworks demonstrate that externalities matter, they typically treat position effects in one dimension or at the level of the co-displayed set, abstracting away the two-dimensional layout of the page and the role of spatial proximity.

In this paper, we introduce a complementary perspective to modeling externalities in online advertising that emphasizes \emph{spatial effects} across ad slots. Specifically, we introduce a framework in which the value of an ad depends not only on its own slot but also on the spatial configuration of neighboring ads on the page -- capturing how proximity, adjacency, and local crowding modulate ad effectiveness. This spatial-competition model departs from existing externality frameworks by making layout and distance first-class primitives. Within this framework, we design simple allocation and auction mechanisms with provable performance guarantees for both welfare and revenue, providing new insights into optimizing slates when ads can cannibalize one another through spatial crowding.

\subsection{Our Model} 
A platform has a webpage with $m$ ad slots. There are $n$ advertisers competing for these ad slots, where advertiser $i$ has a baseline valuation $v_{i,j}$ for ad slot $j$. We assume, without loss of generality, that $n \geq m$, since additional "virtual" advertisers with zero value ($v_{i, j} = 0$) can be introduced as needed. The platform's goal is to allocate ads to ad slots to maximize total social welfare.

If ad slots were independent, so that an advertiser’s value for a slot did not depend on which other slots are filled, the problem would reduce to the standard bipartite matching problem. However, in realistic settings, ads placed in nearby positions compete for user attention, creating negative externalities across the slate. To capture this negative effect while maintaining tractability, we make the following modeling assumptions:

\begin{itemize}
    \item The externality on an ad’s realized value depends only on the set of slots selected for display, and not on which advertiser is assigned to each slot. This reflects the intuition that user attention is shaped primarily by spatial crowding rather than by the identity of advertisers shown in others slots. 
    \item The ad slots are embedded in a metric space -- for example, a two-dimensional plane representing the webpage layout -- where the distance between two slots measures their spatial proximity. This captures the natural idea that the negative impact of crowding diminishes with distance, being stronger for nearby ads and weaker for distant ones.
\end{itemize}

Note that with the above assumptions, without loss we can further set up the objective function of our model as follows: let $M \subseteq [n] \times [m]$ denote a matching between advertisers and ad slots, where $(i, j) \in M$ represents that advertiser $i$ is assigned to ad slot $j$. Let $S(M) = \{j | (i, j) \in M\}$ be the set of slots occupied under matching $M$. Our social welfare objective can be defined as
\begin{align}
    \sw(M) ~:=~ \sum_{(i,j) \in M} v_{i,j} \cdot \delta(j, S(M)), \label{eq:sw}
\end{align}
where function $\delta(j, S(M))$ is a discount function capturing how the presence of other occupied slots $S \setminus \{j\}$ suppresses the value of placing advertiser $i$ at slot $j$.  Then, there must be $\delta(j, S) \in [0, 1]$, as it reflects a discount on each advertisement's value, and $\delta(j, S)$ should be a non-increasing function of $S$, as intuitively the more ads are placed in other slots, the greater the negative impact is on slot $j$. Furthermore, let $d(\cdot, \cdot): [m] \cdot [m] \to [0, 1]$ be the distance function of the metric space for ad slots $[m]$, where we rescale the metric to assume that the diameter (the maximum possible distance between any two slots) of the metric space is $1$. Then, $\delta(j, S)$ should only depend on the values of $d(j, j')$ for $j' \in S \setminus \{j\}$.


We remark that our use of a \emph{metric space} to build the objective function of the model enables a much broader modeling scope that goes beyond the physical distance between slots, as it may encode any notion of \emph{similarity or substitutability} between slots. For instance, two distant slots on a webpage may serve similar user contexts, leading to high competition for user attention. By properly establishing a suitable metric embedding, we can capture the phenomenon that placing ads simultaneously in these two similar slots induces significant negative externalities. Generally speaking, as the metric can capture semantic similarity, functional overlap, or even competition for the same user intent, this flexibility makes our model a unifying abstraction that can accommodate diverse sources of inter-ad interference in online advertising.

With this flexibility in mind, the next key modeling decision is how to specify the discount function $\delta(j, S)$. Different forms of $\delta$ capture different intuitions about how crowding affects ad effectiveness. We introduce two natural instantiations that serve as canonical examples for our framework.

\paragraph{Nearest-Neighbor model.}
We define the discount factor at slot $j$ under occupied set $S$ by its distance to the closest other occupied slot:
\[
\delta_{\min}(j,S)\;:=\;
\begin{cases}
1, & \text{if } S=\{j\},\\[6pt]
\displaystyle \min_{j'\in S\setminus\{j\}} d(j,j'), & \text{otherwise.}
\end{cases}
\]
This captures the intuition that the dominant interference comes from the nearest competing ad, while more distant ads have negligible effect. 
Since all distances $d(j,j')$ are in $[0,1]$, it follows that $\delta_{\min}(j,S)\in[0,1]$, and the monotonicity of $\delta_{\min}$ is guaranteed by the non-increasing property of the minimum operator. 
Analytically, the Nearest-Neighbor model is simpler and highlights local crowding effects.

\paragraph{Product-Distance model.}
We define the discount factor at slot $j$ under occupied set $S$ by the product of pairwise distances:
\[
\delta_{\mul}(j,S)\;:=\;
\begin{cases}
1, & \text{if } S=\{j\},\\[4pt]
\displaystyle \prod_{j'\in S\setminus\{j\}} d(j,j'), & \text{otherwise.}
\end{cases}
\]
Here, each additional occupied slot $j'$ imposes an independent attenuation $d(j,j')$ on the effective value at $j$, and the total discount is the product of all such effects. 
As $d(j, j') \in [0, 1]$, we have $\delta_{\mul}(j,S)$ lies in $[0,1]$, and  is monotone non-increasing in the set $S$.
This model reflects a more global view of spatial competition, where every occupied slot contributes to the reduction in value.

\paragraph{Remark.} We remark that the Nearest-Neighbor model and the Product-Distance model capture the two most natural extremes: one where interference is dominated by the closest neighbor, and another where every neighbor contributes multiplicatively. Other formulations are certainly possible—for example, weighted combinations or more complex attention-decay functions—but we leave them to future work. Moreover, as we show later, even the Product-Distance model already leads to significant analytical challenges, which makes studying more general forms of $\delta$ intractable at present.

\paragraph{Allocation and auction.} Finally, we note that our framework naturally gives rise to two distinct versions of the problem that we analyze in this paper. The first is the \emph{allocation} version, where all valuations $v_{i,j}$ are known to the mechanism and the task is to compute a matching that maximizes the social welfare objective defined above. The second is the \emph{auction} version, where advertisers’ valuations are private inputs, and the mechanism must elicit them truthfully while ensuring that the resulting allocation is incentive compatible. Our results will address both versions, highlighting the algorithmic and mechanism-design aspects of ad allocation under spatial externalities.

\subsection{Our Results}

In this paper, we study the allocation and auction problems under the two spatial models introduced above: the Nearest-Neighbor model and the Product-Distance model. Our objective is to understand the best outcomes that can be achieved in polynomial time. Since computing the exact optimal allocation may require exponential time, we measure the effectiveness of algorithms by their approximation ratio with respect to the true optimum. Throughout, we present both positive and negative results for these models, highlighting the gap between what is information-theoretically possible and what can be achieved efficiently in practice.

\paragraph{Nearest-Neighbor model: a constant approximation.} We first provide our main result for the Nearest-Neighbor model. The work most closely related to our Nearest-Neighbor model is  \cite{BGMS-NeurIPS16}, which aimed to find diverse elements (slots in our model) in metric spaces. It maps to the allocation problem in our model where all valuations $v_{i,j}$ are equal to $1$.
They established two key results: first, assuming the planted clique conjecture, this restricted problem admits no polynomial-time algorithm with approximation ratio better than $2$; and second, they proposed an LP relaxation that achieves a $2e$-approximation. Together, these results imply that for the Nearest-Neighbor model, constant-factor approximations are the best one can hope for.

Our first result for this model is a polynomial-time algorithm that achieves a constant approximation ratio. Moreover, the allocation rule is monotone and therefore can be converted into a truthful mechanism:


\begin{Result}[Informal \Cref{thm:nn-alloc-constant} and \Cref{thm:nn-auc}]
\label{res:constant}
    For Nearest-Neighbor model, there exists a polynomial time algorithm that achieves an $18$-approximation against the optimal social welfare. Furthermore, the algorithm is monotone in the bids of each advertiser and can be converted into a truthful mechanism.
\end{Result}

The main idea behind our approach is to generalize the LP relaxation of \cite{BGMS-NeurIPS16} from their simplified setting to the full Nearest-Neighbor matching model, and provide an efficient rounding algorithm for the optimal LP solution. 

An additional challenge for our setting is that the rounding algorithm in \cite{BGMS-NeurIPS16} does not lead to a truthful auction. To design a truthful auction, we carefully construct the probabilities in our rounding algorithm so that it can be further transformed into a truthful mechanism via Myerson’s Lemma \cite{Myerson1981}.



\paragraph{Stronger results for Nearest-Neighbor model under additional assumptions.} 
For the most general setting of the Nearest-Neighbor model, we established a constant-factor approximation algorithm. While this resolves the question of whether constant approximations are possible, the result has two limitations. 

The first limitation is about the simplicity and efficiency of the algorithm. Our algorithm is based on solving a linear programming relaxation, which makes it relatively slow in an auction environment: the LP needs to be solved polynomially many times to compute each advertiser’s payment. However, in reality, when the advertisers arrive with their bids,  a simpler algorithm is more appreciated to compute the matching between advertisers and slots together with the payments in a shorter time. Inspired by this, we present a simpler algorithm  Nearest-Neighbor model under an extra assumption that the advertisers' valuations can be factorized:

\begin{Result}[Informal \Cref{thm:simple-general} and \Cref{thm:simple-stochastic}]
\label{res:simultaneous}
    For Nearest-Neighbor modTel, assuming each advertiser’s valuation factorizes as $w_i \cdot u_j$, where $w_i$ is a private parameter of advertiser $i$ and $u_j$ is a public value associated with slot $j$.  Then, there exists a simple auction algorithm that achieves an $O(\log m)$-approximation. Furthermore, if each $w_i$ is drawn from a publicly known distribution $D_i$, the approximation ratio can be improved to a constant. 
\end{Result}

The auction algorithm we run for \Cref{res:simultaneous} can be decomposed into two steps. Firstly, we select a subset of slots that will ultimately be used for placing ads. 
Since this step does not depend on advertisers’ private information, the subset $S$ can be computed offline in advance. Secondly, once the advertisers’ private values $w_i$ are revealed, we apply a simple greedy allocation: assign the advertisers in decreasing order of $w_i$ to the slot in decreasing order of $u_j \cdot \min_{j'\in S} d(j, j')$. This allocation rule is monotone, and therefore can be converted into a truthful auction algorithm. Furthermore, the simplicity of the algorithm allows it to be efficiently implemented in practice, making it well-suited for online advertising platforms where speed is essential.

The second limitation of our $1/18$-approximation algorithm is that the approximation factor we obtain is a bit far from $1$. When the metric is captured by the Euclidean distance in the 2D plane, we provide a PTAS algorithm, which gives an approximation ratio arbitrarily close to $1$:

\begin{Result}[Informal \Cref{thm:ptas}]
\label{res:ptas}
    For Nearest-Neighbor model, assuming each $v_{i, j} = 1$ and the metric $d (\cdot, \cdot)$ is  given by Euclidean distances in the 2D plane, there exists a polynomial-time approximation scheme (PTAS) allocation algorithm that achieves $(1 - \epsilon)$-approximation against the optimal allocation.
\end{Result}

The main idea of \Cref{res:ptas} follows from the PTAS algorithm  for selecting independent disks in 2D Euclidean space \cite{EJS-SICOMP05}. Furthermore, the existence of this PTAS shows that the when the metric is the Euclidean distance in 2D plane, it is strictly easier than the general metric case -- this follows from the hardness result established in \cite{BGMS-NeurIPS16} ruling out approximations better than~2 under general metrics. 

\paragraph{Product-Distance model: No non-trivial approximation is achievable.} In sharp contrast to the Nearest-Neighbor model, we show that in the Product-Distance model, it is hard to approximate in polynomial time. We prove the hardness by reducing the \textsc{Max-Independent-Set} (MIS) problem to the allocation problem in this model. It is well known that finding the maximum independent set on a graph $G=(V, E)$ does not admit a polynomial time algorithm that gives a $|V|^{1 - \epsilon}$ approximation (e.g., see \cite{Zuckerman-ToC07}). Therefore, our reduction leads to the following:


\begin{Result}[Informal \Cref{thm:pd-hard}]
The allocation problem under the Product-Distance model admits no polynomial time approximation algorithm with ratio $m^{1-\varepsilon}$ for any $\varepsilon > 0$. 
\end{Result}

Observe that if we restrict the allocation to consist of a single ad slot, then the resulting welfare is clearly an $m$-approximation of the optimal allocation in the Product-Distance model, as it is already an $m$-approximation even when the discount factor $\delta_{\mul}(j, S(M))$ is not applied in the objective. Furthermore, the optimal allocation when restricting to only single slot is monotone in the bids of each advertiser, and therefore can be transformed into a truthful auction using standard Myerson's Lemma \cite{Myerson1981}. Therefore, the hardness result together with the trivial $m$-approximation shows that our understanding of the Product-Distance model is essentially tight. 

\subsection{Related Work}

\paragraph{Externalities in online advertising.} Externalities in online advertising have been studied along several complementary dimensions. \cite{GM-WWW08} initiated the study of externalities in online advertising and introduced audience-choice models to formalize how co-displayed ads compete for limited attention, with higher-quality ads siphoning clicks from others. Scan-based models capture sequential browsing behavior: in cascade and Markovian formulations, interaction with early positions depresses engagement with lower slots and breaks simple rank-by-(bid$\times$CTR) optimality \cite{KM-WINE08, AFMP-WINE08}. Empirically, position-bias studies document that click outcomes are shaped by the surrounding slate \cite{CZTR-WSDM08}, and work on visual clutter shows that crowding reduces advertising effectiveness \cite{HRK-JAR09} (see also recent evidence on digital clutter and memory effects by Williams and Lynch, 2024). On the mechanism-design side, externalities within keyword auctions have been analyzed both theoretically and empirically: \cite{GK-WINE08} study equilibria and efficiency of GSP with externalities, while \cite{GIM-WINE09} combine theory with data from a large search engine to quantify their impact. Together, these strands establish that inter-ad effects matter, though most frameworks emphasize sequential position effects or set-level dependencies rather than the spatial layout and proximity effects that motivate our model.

\paragraph{Diverse subsets in metric space and dispersion problem.}The algorithmic problem of finding the optimal allocation in our model is related to the well-studied problem of selecting diverse subsets in metric spaces and the closely related dispersion problems. Both problems aim to select a subset of $k$ facilities from $m$ locations to maximize some diversity objective. \cite{BGMS-NeurIPS16} introduce linear-programming relaxations that yield the first constant-factor approximation for the sum-min diversity objective. \cite{BorodinJLY17} analyze the max-sum diversification objective, providing approximation guarantees under matroid and knapsack constraints. \cite{CevallosEZ19} show that local search achieves a $(1-O(1/k))$-approximation for max-sum diversification over negative-type metrics (and under matroid constraints). For large-scale data, \cite{IndykMMM14} develop composable coresets that enable distributed or streaming approximation for both diversity and coverage objectives. Moreover, \cite{CeccarelloPPU17} design space and round-efficient MapReduce and streaming algorithms in doubling metrics. \cite{CeccarelloPP18} further give approximation schemes specialized to the doubling-dimension setting. 

In the closely related dispersion problem literature, \cite{FeketeM04} consider the objective of the average distance between selected locations. They provide a linear-time algorithm to solve the problem exactly in d-dimensional space when k is a constant, and a polynomial time approximation scheme (PTAS) when k is not fixed. \cite{DumitrescuJ10} obtain improved constant-factor approximation algorithms for dispersion in disks. \cite{Pisinger06} derives strong upper bounds and exact algorithms for the p-dispersion-sum problem, where the goal is to maximize the distance sum between selected locations. More recently, \cite{MishraD21} study the c-dispersion problem in general metrics, which aims to maximize the minimum distance from a point to its nearest $c$ points. They give a simple greedy $2c$-approximation, along with parameterized hardness results.

Unlike the aforementioned works that select a subset of locations (slots), our paper study the allocation problem between advertisers and slots. The added complexity arises because advertisers’ valuations are heterogeneous and, in mechanism design, private.

\paragraph{PTAS in disk graphs.} Our PTAS algorithm in the nearest-neighbor model in the 2D Euclidean space is closely-related to the algorithm for selecting independent disks. In disk (and disk-like) intersection graphs, a PTAS is known to exist in classic problems such as Maximum Weight Independent Set (MWIS) and Minimum Weight Vertex Cover \cite{ErlebachJS05}. For unit disk graphs, \cite{NiebergHK04} designed a robust PTAS for MWIS that does not require an explicit geometric embedding. \cite{Marx05} studied when PTASes can be upgraded to EPTAS and identified obstacles for such improvements on geometric problems, including hardness for MIS in unit disk and unit square graphs. Strengthening and systematizing these directions, \cite{vanLeeuwen06} build a general framework for designing EPTASes for vertex-deletion problems on disk graphs, including Vertex Cover. For more general families beyond disks, \cite{ChanHP12} obtained a PTAS for MIS of pseudo-disks via local search and a constant-factor approximation for the MWIS using an LP rounding scheme.

\section{Constant Approximation for Nearest-Neighbor Model}
\label{sec:LP}

In this section, we present our algorithm  for the Nearest-Neighbor model and show \Cref{res:constant}.

\subsection{Constant Approximation Allocation}

We first show that there exists an allocation algorithm that achieves a constant approximation for the Nearest-Neighbor model. To be specific, we prove the following theorem:
\begin{Theorem}
\label{thm:nn-alloc-constant}
    For Nearest-Neighbor model, there exists a polynomial time algorithm that achieves $18$-approximation against the optimal social welfare.
\end{Theorem}

To begin, our first step is to reduce the problem to a geometric formulation. Formally, we define the following:

\begin{Definition}[Grouped Pseudo-Disk Selection Problem]
\label{def:gpds}
Let $(\X,d)$ be a metric space with $1$ being an upper bound of its diameter. 
A \emph{pseudo-disk} is a triple $(c,r,w)$ consisting of a center $c \in \X$, a radius $r \in [0,1]$, and a weight $w \geq 0$. 
A pseudo-disk $D = (c,r,w)$ corresponds to the set of points in $\X$ with distance \emph{strictly smaller} than $r$ from $c$, and is associated with weight $w$. 

Given a set of pseudo-disks partitioned into $K$ groups $T_1,T_2,\dots,T_K$. 
The \emph{Grouped Pseudo-Disk Selection Problem} asks to find a selection $T$ that chooses \emph{at most one} pseudo-disk from each group with the objective of maximizing
\[
\val(T) ~:=~ \sum_{D = (c_D, r_D, w_D) \in T} w_D,
\]
 such that one of the following constraints is satisfied:
\begin{constraints}
    \item \label{cons:I} Constraint (I): Every center $c \in \X$ in the metric space is covered by at most one selected disk.
    \item \label{cons:II} Constraint (II): The center of each selected disk is not covered by another selected disk.
\end{constraints}
\end{Definition}

The Grouped Pseudo-Disk Selection problem defined in \Cref{def:gpds} is not fully settled, as the problem may switch between \Cref{cons:I} and \Cref{cons:II}. Clearly, \Cref{cons:I} is stronger than \Cref{cons:II}, as every selection satisfying \Cref{cons:I} must satisfy \Cref{cons:II}. We remark that both constraints are needed in our reduction, as we show the following:

\begin{Lemma}
\label{lma:reduction}
The allocation problem under the Nearest-Neighbor model with $n$ advertisers and $m$ ad slots admits a polynomial-time reduction to a Grouped Pseudo-Disk Selection instance with $K = n$ groups of pseudo-disk and each group $T_k$ has size $|T_k| \leq m^2$, such that the instance is a $2$-approximation of the original problem. To be specific, the reduced instance satisfied the following:
\begin{enumerate}
    \item Let $M^*$ be the optimal matching of the Nearest-Neighbor instance and $T^* \subseteq \bigcup_{k \in [K]} T_k$ be the optimal selection of the Grouped Pseudo-Disk Selection instance that satisfies the stronger \emph{\textbf{\Cref{cons:I}}}. Then, $\sw(M^*) \leq \val(T^*)$.
    \item Given a selection $T$ of the Grouped Pseudo-Disk Selection instance that satisfies the weaker \emph{\textbf{\Cref{cons:II}}}, we can find a feasible matching $M$ for the Nearest-Neighbor instance in polynomial time, such that  $\sw(M) \geq 0.5 \cdot \val(T)$.
\end{enumerate}
\end{Lemma}

We defer the proof of \Cref{lma:reduction} to \Cref{sec:LP_Missing}. With the help of \Cref{lma:reduction}, to give a constant approximation, it remains to find a good selection for a Grouped Pseudo-Disk Selection instance that satisfies \Cref{cons:II}. To be specific, we show the following:

\begin{Lemma}
    \label{lma:nn-alloc-approx}
    For a Grouped Pseudo-Disk Selection instance, let $T^*$ be the optimal selection that satisfies \emph{\textbf{\Cref{cons:I}}}. There exists a polynomial time algorithm that finds a random selection $\widetilde T$, such that 
    \begin{itemize}
        \item $\widetilde T$ satisfies \emph{\textbf{\Cref{cons:II}}}.
        \item $\E[\val(\widetilde T)] \geq \frac{1}{9} \cdot \val(T^*)$.
    \end{itemize}
\end{Lemma}

We defer the proof of \Cref{lma:nn-alloc-approx} to \Cref{sec:LPRelax}, and first prove \Cref{thm:nn-alloc-constant} via combining \Cref{lma:reduction} and \Cref{lma:nn-alloc-approx}:

\begin{proof}[Proof of \Cref{thm:nn-alloc-constant}]
    Given a Nearest-Neighbor instance, we first apply \Cref{lma:reduction} to reduce the problem to a Grouped Pseudo-Disk Selection instance. Then, we run the algorithm provided by \Cref{lma:nn-alloc-approx} to find $\widetilde T$, and use Property (2) of \Cref{lma:reduction} to convert it back to a feasible solution $\widetilde M$. By Lemmas \ref{lma:reduction} and \ref{lma:nn-alloc-approx}, we have
    \[
    \E[\sw(\widetilde M)] ~\geq~ 0.5 \cdot \E[\sw(\widetilde T)] ~\geq~ \frac{1}{18} \cdot \val(T^*) ~\geq~ \frac{1}{18} \cdot \sw(M^*),
    \]
    where $M^*$ is the optimal solution for the Nearest-Neighbor instance, and $T^*$ is the optimal selection of the Grouped Pseudo-Disk Selection instance that satisfies \Cref{cons:I}.
\end{proof}

\subsection{Proof of \Cref{lma:nn-alloc-approx} via LP Relaxation}
\label{sec:LPRelax}

In this subsection, we provide the proof of \Cref{lma:nn-alloc-approx}. Our main idea is to relax the Grouped Pseudo-Disk Selection problem with \Cref{cons:I} to an LP, and further round the LP to give a solution that satisfies \Cref{cons:II}.

Let $D_1, \cdots, D_N$ be the pseudo disks, where $D_i = (c_i, r_i, w_i)$. Let $T_1, \cdots, T_k$ be the group of disks, such that each $i \in [N]$ belongs to exactly one $T_k$. Consider the following LP:

\begin{align}
    \text{maximize } \quad \quad  &     \quad \sum_{i \in [N]} w_i \cdot x_i,\notag \\
    \text{s.t.}\qquad \forall c \in \X &   \quad   \sum_{i \in [N]: d(v, c_i) < r_i} x_i \leq  1\tag{$\text{LP}_{\textsc{relax}}$} \label{lp:relax}  \\
    \forall k \in [K] &    \quad \sum_{i \in [N]: i \in T_k} x_i \leq 1 \notag  \\
    \forall i \in [N] &    \quad x_i \in [0, 1] \notag
\end{align}

\begin{Lemma}
    \label{lma:LPRelax}
    \ref{lp:relax} is a feasible relaxation of the Grouped Pseudo-Disk Selection problem with \Cref{cons:I}, that is, given a Grouped Pseudo-Disk Selection instance, the objective of \ref{lp:relax} is at least $\val(T^*)$, where $T^*$ is the optimal selection that satisfies \Cref{cons:I}.
\end{Lemma}

\begin{proof}
    It's sufficient to show that solution $\{\widetilde x_i\}_{i \in [N]}$ with $\widetilde x_i = \one[i \in T^*]$ is a feasible integral solution of \ref{lp:relax}. For the first constraint, since $T^*$ is feasible with \Cref{cons:I}, for every $c \in \X$, $\sum_{i: d(v, c_i) } \widetilde x_i \leq 1$ holds. For the second constraint, since $T^*$ selects at most one disk from each group $T_k$, $\sum_{i \in T_k} \widetilde x_i \leq 1$ is satisfied. Therefore, $\widetilde x_i = \one[i \in T^*]$ is a feasible integral solution of \ref{lp:relax}.
\end{proof}

With \ref{lp:relax}, we now present our main \Cref{alg:rounding}, which is the desired algorithm for \Cref{lma:nn-alloc-approx}.

\begin{algorithm}[tbh]
\caption{\textsc{Constant Approximation for Group Pseudo-Disk Selection}}
\label{alg:rounding}
\KwIn{Disks $D_1, \cdots, D_N$ and groups $T_1, \cdots, T_k$}
Solve \ref{lp:relax}. Let $x^*_i$ be the optimal solution.\\
Initiate $T = \varnothing$ and $\widetilde T = \varnothing$. \\
For each $i \in N$, add $i$ into $T$ with probability $x^*_i/3$. \\
\For{$i \in T$}
{
Let $k\in [K]$ be the group that satisfies $i \in T_k$. \\
\If{$\nexists~ i' \in T: i' \neq i \land (i' \in T_k \lor d(c_i, c_{i'}) < r_{i'})$  }
{
Add $i$ into $\widetilde T$ with probability
\begin{align}
    \frac{1}{3} \cdot \frac{1}{\prod_{i' \neq i} \left(1 - \one[i' \in T_k \lor d(c_i, c_{i'}) < r_{i'}] \cdot x^*_{i'}/3\right)}. \label{eq:alg-prob}
\end{align}
}
}
\KwOut{Feasible selection $\widetilde T$}
\end{algorithm}

To show \Cref{alg:rounding} is the desired algorithm, we first show that \Cref{alg:rounding} is a feasible algorithm, and then show that the expectation of $\val(\widetilde T)$ is at least $\frac{1}{9} \cdot \val(T^*)$.

\paragraph{Feasibility of \Cref{alg:rounding}.} To show \Cref{alg:rounding} is feasible, we need to show
\begin{itemize}
    \item The probability defined in \eqref{eq:alg-prob} is a feasible probability that falls in $[0, 1]$.
    \item The output $\widetilde T$ satisfies \Cref{cons:II}.
\end{itemize}
For the feasibility of \eqref{eq:alg-prob}, clearly the probability in \eqref{eq:alg-prob} must be non-negative. On the other hand, note that
\begin{align*}
    &\prod_{i' \neq i} \left(1 - \one[i' \in T_k \lor d(c_i, c_{i'}) < r_{i'}] \cdot x^*_{i'}/3\right) \\
    ~\geq~& 1 - \sum_{i' \in T_k \setminus \{i\}} \frac{x^*_{i'}}{3} - \sum_{i' \neq i: d(c_i, c_{i'} < r_{i'}} \frac{x^*_{i'}}{3} \\
    ~\geq~& 1 - \frac{1}{3} - \frac{1}{3} ~\geq~ \frac{1}{3},
\end{align*}
where the first inequality follows from union bound, and the second inequality follows from the constraints in \eqref{lp:relax}. Therefore, the probability in \eqref{eq:alg-prob} is upper-bounded by $1$.

For the feasibility of $\widetilde T$, note that a disk $i$ is possible to be added into $\widetilde T$ only when it's in $T$, and its center is not covered by another disk in $T$. Therefore, the center of every disk in $\widetilde T$ is not covered by a second disk, and therefore \Cref{cons:II} is satisfied.

\paragraph{The value of $\widetilde T$.} Now, we show $\E[\val(\widetilde T)] \geq \frac{1}{9} \cdot  \val(T^*)$. To achieve this, we show a stronger claim that
\begin{align}
    \E[\val(\widetilde T)] ~=~ \frac{1}{9} \cdot \text{objective of \eqref{lp:relax}} ~\geq~ \frac{1}{9}\cdot  \val(T^*), \label{eq:equal}
\end{align}
where the inequality is guaranteed by \Cref{lma:LPRelax}. To show the equality in \eqref{eq:equal}, note that for each disk $i$, we have
\begin{align*}
    \pr[i \in \widetilde T] &= \pr[i \in T] \cdot \pr[\nexists i'\neq i \in T: i' \in T_k \lor d(c_i, c_{i'}) < r_{i'}] \cdot \eqref{eq:alg-prob} \\
    &= \frac{x^*_i}{3} \cdot \frac{1}{3} \cdot \frac{\prod_{i' \neq i} \left(1 - \one[i' \in T_k \lor d(c_i, c_{i'}) < r_{i'}] \cdot x^*_{i'}/3\right)}{\prod_{i' \neq i} \left(1 - \one[i' \in T_k \lor d(c_i, c_{i'}) < r_{i'}] \cdot x^*_{i'}/3\right)} \\
    &= \frac{x^*_i}{9},
\end{align*}
where the first equality follows from the fact that whether $i \in T$, whether $i$ satisfies the constraint in Line 6 of \Cref{alg:rounding}, and whether $i$ is added into $\widetilde T$ in Line 7 of \Cref{alg:rounding} are independent events. By multiplying $w_i$ to both sides of the above equality and summing over all disks $i \in [N]$, we have $\E[\val(\widetilde T)]$ equals to $\frac{1}{9}$ times the objective of \eqref{lp:relax}.

\paragraph{Efficiency of \Cref{alg:rounding}.} Finally, we need to show \Cref{alg:rounding} is efficient. Since there are $O(|\X| + K + N)$ constraints and $O(N)$ variables in \eqref{lp:relax}, it requires $O(\poly(|\X|, K, N))$ time to solve \eqref{lp:relax}. For the remaining steps in \Cref{alg:rounding}, it can be clearly verified that these steps run in $O(\poly(|\X|, K, N))$ time. Therefore, \Cref{alg:rounding} is efficient.

\subsection{Converting the Algorithm to a Truthful Mechanism}

Now, we show how to convert the allocation we discussed above to a truthful mechanism. More specifically, we prove that the allocation algorithm in \Cref{thm:nn-alloc-constant} is monotone, so that the Myerson's Lemma can be applied. We defer the proof of \Cref{thm:nn-auc} to \Cref{sec:LP_Missing}.

\begin{Theorem}
\label{thm:nn-auc}
    The allocation algorithm for Nearest-Neighbor model for \Cref{thm:nn-alloc-constant} can be converted into a truthful mechanism.
\end{Theorem}

\section{Simple Algorithms for Nearest-Neighbor Model}
\label{sec:simple}

In this section, we present efficient and simple algorithms for the special case of the Nearest-Neighbor model, assuming value $v_{i, j}$ can be factorized as $v_{i, j} = w_i\cdot u_j$, where $u_j$ is publicly known. Similar to \Cref{thm:nn-auc}, we also show that the algorithms can be converted into truthful mechanisms.

The main intuition for algorithms in this section is \emph{simultaneous optimization}. The concept of simultaneous optimization is defined for a combinatorial optimization problem, where the objective is to optimize an unknown norm. The goal of simultaneous optimization is to provide a feasible solution for the combinatorial optimization problem, such that the solution is approximately optimal for every norm in a certain class. We show that this \emph{factorized} Nearest-Neighbor model with the assumption that $v_{i, j} = w_i\cdot u_j$ can be modeled as a simultaneous optimization problem over a class of norms parameterized by the unknown advertisers' values $\{w_i\}$. When the values $\{w_i\}$ are arbitrary, we present a direct proof that gives an $O(\log m)$ approximation. When the value $\{w_i\}$ are stochastic such that $w_i \sim \D_i$, we show that it's sufficient to use the expectations of each distribution as parameters that define the norm of the problem, and the optimal selection of slots under this specific norm achieves a constant approximation.

\subsection{$O(\log m)$ Approximation for Arbitrary Advertiser Values}
\label{sec:simple-general}

We first give an $O(\log m)$ approximation algorithm when advertisers' values are adversarial. To be specific, we prove the following:

\begin{Theorem}
\label{thm:simple-general}
    For Nearest-Neighbor model, assuming each advertiser’s valuation factorizes as $w_i \cdot u_j$, where $w_i$ is a private value of advertiser $i$ and $u_j$ is a public weight associated with slot $j$. Then, there exists an allocation algorithm running in $\OTild(n + m)$ time that achieves an $O(\log m)$-approximation. Furthermore, the algorithm can be converted to a truthful mechanism.
\end{Theorem}

To prove \Cref{thm:simple-general}, we show \Cref{alg:simple-general} is the desired algorithm for \Cref{thm:simple-general}.  Throughout the subsection, we assume the number of slots $m$ is at most the number of advertisers $n$, i.e., $m \leq n$. This is without loss of generality, as we may arbitrarily add dummy advertisers with value $0$ into the instance.

\begin{algorithm}[tbh]
\caption{\textsc{$O(\log m)$ Approx. for Factorized Nearest-Neighbor Model}}
\label{alg:simple-general}
\KwIn{Number of slots $m$ and advertisers $n$, distance function $d$}
\tcc{Slot Selection Phase}
Input slot weights $u_1, u_2, \cdots, u_m$, such that WLOG $u_1 \geq u_2 \geq \cdots \geq u_m$. \\
Let $L = \lceil \log_2(m^2) \rceil $. \\
Choose $r \in \{2^{-0}, 2^{-1}, 2^{-2}, \cdots, 2^{-L}\}$ uniformly at random. \\
Initiate slot set $\widetilde S = \varnothing$. \\
\For{$j = 1 \to m$}
{
\If {$d(j, j') \geq r$ holds for every $j' \in \widetilde S$}
{
Add $j$ into $\widetilde S$
}
}
\tcc{Allocation Phase}
Advertisers report values $w_1, w_2, \cdots, w_n$.\\
For $i \in |\widetilde S|$, assign the advertiser with $i$-th largest value to the slot in $\widetilde S$ such that the value of $u_j \cdot \min_{j' \in \widetilde S \setminus \{j\}} d(j, j')$ is the $i$-largest.  \\
Let $\widetilde M$ be the corresponding matching. \\
\KwOut{Matching $\widetilde M$.}
\end{algorithm}





\paragraph{Representing the objective as a norm.} To show that  \Cref{alg:simple-general} achieves an $O(\log m)$ approximation, we first introduce some extra notations that will be used in our proof.

Given an instance of the factorized Nearest-Neighbor model where $w_1, w_2, \cdots , w_n$ are values of advertisers and $u_1 \geq u_2 \geq \cdots \geq u_m \geq 0$ are weights of slots. We define
\[
W ~:= (w_1, \cdots, w_n)
\]
to be the value vector of advertisers. 

For an $m$-dimensional vector $Z = (z_1, z_2, \cdots, z_m) \in \R^m_{\geq 0}$ that only contains non-negative values, we define notation
\[
Z^{\downarrow} ~:=~ (z^\downarrow_1, \cdots, z^\downarrow_m)
\]
to be a vector denoting the sequence obtained by ordering values in $Z$ in non-increasing order, where notation $z^\downarrow_i$ represents the $i$-th largest value in $Z$. Then, for $Z \in \R^m_{\geq 0}$, we define function $f: \R^m_{\geq 0} \to \R_{\geq 0}$ parametrized by $W^\downarrow$ as follows (recall that $n \geq m$):
\begin{align*}
    f(Z; W^\downarrow) ~:=~ \sum_{j = 1}^{m} z^\downarrow_j \cdot w^\downarrow_j.
\end{align*}

Then, the following claim shows that the function $f$ is a norm:

\begin{Claim}
\label{clm:f-norm}
    Function $f(Z, W^\downarrow)$ is a monotone symmetric norm. 
\end{Claim}
\begin{proof}
    The symmetricity of function $f$ follows from the fact that $f$ is applied to the sorted vector $Z^\downarrow$. Therefore, the order of values in $Z$ does not effect the value of function $f$. For the monotonicity, consider two $n$-dimensional vectors $(z_1, \cdots, z_m)$ and $(y_1, \cdots, y_m)$ that satisfies $z_i \geq y_i$.Then, there must be $z^\downarrow_i \geq y^\downarrow_i$: there are at least $i$ values in vector $y$ being at least $y^\downarrow_i$, and therefore there are at least $i$ values in vector $z$ being at least $y^\downarrow_i$, which implies $z^\downarrow_i \geq y^\downarrow_i$. Then, we have
    \[
    f(Z; W^\downarrow) - f(Y; W^\downarrow) ~=~ \sum_{i \in [m]} (z_i - y_i) \cdot w^\downarrow_i ~\geq~ 0,
    \]
    and therefore the monotonicity holds.

    It remains to show that function $f$ is a norm. We have the following:
    \begin{itemize}
        \item Positivity: The non-negativity of $f$ is guaranteed by the non-negativity of vector $Z$ and the vector $W^\downarrow$. We further assume that $w^\downarrow_1 > 0$, as otherwise it implies all advertisers have value $0$, making the problem trivial. Then, $f(Z, W^\downarrow) \geq z^\downarrow_1 \cdot w^\downarrow_1$, and therefore $f(Z, W^\downarrow) = 0$ implies vector $Z$ only contains value $0$.
        \item Absolute Homogeneity: For any $\alpha \geq 0$, we have
        \[
        f(\alpha\cdot Z; W^\downarrow) ~=~ \sum_{i = 1}^m (\alpha z^\downarrow_i) \cdot w^\downarrow_i ~=~ \alpha \cdot f(Z, W^\downarrow).
        \]
        \item Triangle Inequality: Let $(z_1, \cdots, z_m)$ and $(y_1, \cdots, y_m)$ be two $n$-dimensional vector that are not necessarily ordered. As function $f$ is symmetric, we assume without loss of generality that $z_1 + y_1 \geq z_2 + y_2 \geq \cdots \geq z_m + y_m$. We assume without loss of generality that $w^\downarrow_{m+1}:= 0$, as $w^\downarrow_{m+1}$ is not needed in function $f$. Since $w^\downarrow_i \geq w^\downarrow_{i+1}$ for every $i \in [n]$, we have
        \begin{align*}
            f(Y + Z; W^\downarrow) &= \sum_{i = 1}^m (y_i + z_i) \cdot w^\downarrow_i \\
            &= \sum_{i = 1}^m \left((w^\downarrow_i - w^\downarrow_{i+1}) \cdot \sum_{k = 1}^i (y_k + z_k) \right)\\
            &= \sum_{i = 1}^m \left((w^\downarrow_i - w^\downarrow_{i+1}) \cdot \sum_{k = 1}^i y_k \right)+ \sum_{i = 1}^m \left((w^\downarrow_i - w^\downarrow_{i+1}) \cdot \sum_{k = 1}^i z_k \right) \\
            &= \sum_{i = 1}^m y_i \cdot w^\downarrow_i + \sum_{i = 1}^m z_i \cdot w^\downarrow_i \\
            &\leq f(Y; W^\downarrow) + f(Z; W^\downarrow),
        \end{align*}
        where the inequality holds as it becomes an equality only when vector $Y$ and $Z$ are both sorted.
    \end{itemize}
    Therefore, function $f$ is a norm.
\end{proof}

Now, we show how to represent the objectives of the optimal matching and the matching obtained by \Cref{alg:simple-general} as a function of $f$. For any subset $S \subseteq [m]$ of slots, we define vectors $U(S)$ and $Z(S)$, such that
\[
\big(U(S)\big)_j ~:=~ u_j \cdot \one[j \in S] ~~ \text{and} ~~  \big(Z(S)\big)_j ~:=~ \one[j \in S] \cdot u_j \cdot \min_{j' \in S \setminus \{j\}} d(j, j'),
\]
where notation $\one[\cdot]$ represents the indicator function,
and similarly $U(S), Z(S)$ to be vectors denoting the sequence obtained by ordering values in $U(S)$ and $Z(S)$ in non-increasing order, respectively.

Let $M^*$ be the optimal matching of the factorized Nearest-Neighbor instance. With the above notations, we can simplify the value of $\sw(M^*)$ as
\begin{align}
    \sw(M^*) ~=~ \max_{S \subseteq [m]} f(Z(S); W^\downarrow). \label{eq:optimal-matching-simplify}
\end{align}

On the other hand, For $l \in \{0, 1, 2, \cdots, L = \lceil \log_2(m^2) \rceil \}$, we let $\widetilde S_l$ be the set of slots selected in the slot selection phase of \Cref{alg:simple-general} and $\widetilde M_l$ be the corresponding matching, assuming the parameter $r$ from Line 3 of \Cref{alg:simple-general} equals to $2^{-l}$. Then, we have
\begin{align}
    \sw(\widetilde M_l) ~=~ f(Z(\widetilde S_l); W^\downarrow). \label{eq:alg-matching-simplify}
\end{align}

\paragraph{Proof of the approximation ratio.}
Since the parameter $r$ in \Cref{alg:simple-general} is chosen uniformly from $\{0, 1, \cdots, L \}$, the expected performance of \Cref{alg:simple-general} is exactly $\sum_{l = 0}^L \sw(\widetilde M_l)$. Since $L = \lceil \log_2(m^2) \rceil  = O(\log m)$, to show that \Cref{alg:simple-general} achieves an $O(\log m)$ approximation, it's sufficient to show 
\begin{align*}
    \E[\sw(\widetilde M)] ~=~ \frac{1}{L+1} \cdot \sum_{l = 0}^L \sw(\widetilde M_l) ~\geq~ \frac{1}{L+1} \cdot \frac{1}{4} \cdot  \sw(M^*). 
\end{align*}
By applying \eqref{eq:optimal-matching-simplify} and \eqref{eq:alg-matching-simplify} to the above inequality, it's equivalent to show that for any $S \subseteq [m]$, we have
\begin{align}
    4\cdot \sum_{l = 0}^L f(Z(\widetilde S_l); W^\downarrow) ~\geq~ f(Z(S); W^\downarrow). \label{eq:logm-approx}
\end{align}

Fix the subset $S$ on the RHS of \eqref{eq:logm-approx}. We first partition $S$ into
\[
S = \left(\bigcup_{l = 0}^{L - 1} S_l \right) \cup S_L,
\]
such that for $l \in \{1, 2, \cdots, L-1\}$, we have
\[
 S_l := \left\{j \in S: \min_{j' \in S \setminus \{j\}} d(j, j') \in (2^{-l}, 2^{-l+1}] \right\},
\]
and $S_L := S \setminus \bigcup_{l = 1}^{L - 1} S_l$ includes all remaining items $j$ that satisfies $\min_{j' \in S \setminus \{j\}} d(j, j') \leq 2^{-L+1}$.

For $j \in S_l$, as the value of $\min_{j' \in S \setminus \{j\}} d(j, j')$ is upper-bounded by $2^{-l+1}$, its contribution to $\sw(M^*)$ can be bounded by $u_j \cdot 2^{-l+1}$ times the value of the advertiser matched to $j$. Applying this inequality to the function $f(Z(S); W^\downarrow)$, we have
\begin{align}
    f(Z(S); W^\downarrow) ~\leq~ \sum_{l = 1}^L f(U(S_l); W^\downarrow) \cdot 2^{-l+1}. \label{eq:RHS-upper}
\end{align}

Similarly, for $\widetilde S_l$ obtained by \Cref{alg:simple-general}, note that \Cref{alg:simple-general} adds $j$ into $\widetilde S_l$ only when $j$ is at least $r = 2^{-l}$ away from existing slots in $\widetilde S_l$. Therefore, for $j \in \widetilde S_l$, there must be $\min_{j \in \widetilde S_l \setminus \{j\}} d(j, j') \geq 2^{-l}$, which further implies 
\begin{align}
    f(Z(\widetilde S_l); W^\downarrow) ~\geq~  f(U(\widetilde S_l); W^\downarrow) \cdot 2^{-l}\label{eq:LHS-lower}
\end{align}

Then, the following claim, which serves as the key idea to prove \eqref{eq:logm-approx}, relates the value of $f(U(S_l); W^\downarrow)$ to the value of $f(U(S_{l+1}); W^\downarrow)$:

\begin{Claim}
\label{clm:connect}
    For $l \in \{1, 2, \cdots, L - 1\}$, we have
    \[
    f(U(S_l); W^\downarrow) ~\leq~ 2 \cdot f(U(\widetilde S_{l+1}); W^\downarrow).
    \]
\end{Claim}

\begin{proof}
    For subset $S_l$, we partition it into $S^{(1)}_l \cup S^{(2)}_l$, such that $S^{(1)}_l = S_l \cap \widetilde S_{l+1}$, and $S^{(2)}_l = S_l \setminus \widetilde S_{l+1}$. We further define vector $U_1 = U(S^{(1)}_l)$, and $U_2 = U(S^{(2)}_l)$. 
    
    Since $S^{(1)}_l \subseteq U(\widetilde S_{l+1})$, the vector $U_1$ is dominated by $U(\widetilde S_{l+1})$. By monotonicity of function $f$, this implies
    \[
    f(U_1; W^\downarrow) ~\leq~ f(U(\widetilde S_{l+1}); W^\downarrow).
    \]

    Next, we show $f(U_2; W^\downarrow)$ is also upper-bounded by $f(U(\widetilde S_{l+1}); W^\downarrow)$. We achieve this by building up a mapping from $S^{(2)}_l$ to $\widetilde S_{l+1}$. For $j \in S^{(2)}_l$, since we have $j \notin \widetilde S_{l+1}$, there must exist $j' \in \widetilde S_{l+1}$ that satisfies $d(j, j') < r = 2^{-l-1}$. Furthermore, $j'$ should satisfy $u_{j'} \geq u_j$, as \Cref{alg:simple-general} considers each slot in decreasing order of $u_j$. In this case, we call $j \in S^{(2)}_l$ is ``blocked'' by $j' \in \widetilde S_{l+1}$.

    The key observation of our proof is that for $j_1, j_2 \in S^{(2)}_l$ such that $j_1 \neq j_2$, let $j'_1, j'_2 \in \widetilde S_{l+1}$ be the slots that block $j_1$ and $j_2$ respectively. Then, there must be $j'_1 \neq j'_2$. We prove via contradiction: assuming $j'_1 = j'_2$, which means $d(j_1, j'_1) < 2^{-l-1}$, and $d(j_2, j'_1) < 2^{-l-1}$. By triangle inequality, there must be $d(j_1, j_2) \leq 2 \cdot 2^{-l-1} = 2^{-l}$, which is in contrast to the assumption that $j_1, j_2 \in S_l$, as it implies $d(j_1, j_2) > 2^{-l}$.

    The above observation implies that there exists a mapping from non-zero values in $U_2$ to non-zero values in $U(\widetilde S_{l+1})$, such that each non-zero value in $U_2$ is mapped to a value that is no smaller than itself, and each value in $U(\widetilde S_{l+1})$ is mapped at most once. By rearranging the values in $U_2$ (which is allowed as $f$ is symmetric) and monotonicity of function $f$, we have
    \[
    f(U_2; W^\downarrow) ~\leq~ f(U(\widetilde S_{l+1}); W^\downarrow).
    \]

    Combining the above two inequalities together gives
    \begin{align*}
        f(U(S_l); W^\downarrow) ~\leq~  f(U_1; W^\downarrow) + f(U_2; W^\downarrow) ~\leq~ 2\cdot f(U(\widetilde S_{l+1}); W^\downarrow),
    \end{align*}
    where the first inequality follows from the triangle inequality of function $f$.
\end{proof}

Combining \eqref{eq:RHS-upper}, \eqref{eq:LHS-lower}, and \Cref{clm:connect}, we get
\begin{align*}
    f(Z(S); W^\downarrow) \leq~& \sum_{l = 1}^L f(U(S_l); W^\downarrow) \cdot 2^{-l+1} \\
    \leq~& f(U(S_L); W^\downarrow) \cdot 2^{-L+1} + 2\cdot \sum_{l = 2}^L f(U(\widetilde S_l); W^\downarrow) \cdot 2^{-l+1} \\
    \leq~& f(U(S_L); W^\downarrow) \cdot 2^{-L+1} + 4 \cdot \sum_{l = 2}^L f(Z(\widetilde S_l); W^\downarrow) \\
    \leq~& 4\cdot \sum_{l = 0}^L f(Z(\widetilde S_l); W^\downarrow) \\
    &+ \left(f(U(S_L); W^\downarrow) \cdot 2^{-L+1} - 4\cdot f(Z(\widetilde S_0); W^\downarrow)\right).
\end{align*}

Therefore, to show \eqref{eq:logm-approx}, it only remains to show
\begin{align}
    f(U(S_L); W^\downarrow) \cdot 2^{-L+1} ~\leq~ 4\cdot f(Z(\widetilde S_0); W^\downarrow). \label{eq:edge-case}
\end{align}
By \eqref{eq:LHS-lower}, we have
\[
f(Z(\widetilde S_0); W^\downarrow) ~\geq~ f(U(\widetilde S_0); W^\downarrow) \cdot 1 ~\geq~ u_1\cdot w_1,
\]
where we recall $u_1$ represents the highest slot weight, and $w_1$ represents the highest advertiser's value. On the other hand, we have
\[
f(U(S_L); W^\downarrow) \cdot 2^{-L+1} ~\leq~ 2 \cdot \frac{1}{m^2} \cdot m \cdot (u_1 \cdot w_1) ~\leq~ 2u_1 \cdot w_1,
\]
where the first inequality uses the fact that each value in vector $U(S_L)$ is upper-bounded by $u_1$, and similarly each value in vector $W^\downarrow$ is upper-bounded by $w_1$. Combining the above two inequalities together shows \eqref{eq:edge-case}.

\paragraph{Truthfulness and Efficiency.} Finally, we finish the proof of  \Cref{thm:simple-general} by showing that \Cref{alg:simple-general} is efficient, and can be converted into a truthful mechanism via Myerson's Lemma \cite{Myerson1981}. By \Cref{alg:simple-general}, the allocation phase only requires to sort the advertisers and the slots respectively , and match the advertisers and slots in order. Therefore, the algorithm runs in $\OTild(n + m)$ time.

For the truthfulness of the algorithm, we only need to show that the objective achieved in the allocation phase is  monotone. Fix $\tilde S$ to be the set of pre-selected slots, and let $W = (w_1, w_2, \cdots, w_n)$ be the vector of advertisers' values.  The objective achieved in the allocation phase is exactly
\[
\sum_{i = 1}^m \big(Z^\downarrow(\tilde S) \big)_i \cdot w^\downarrow_i, 
\]
where $w^\downarrow_i$ represents the $i$-th largest value in vector $W$. Now suppose we increase the value of $w_i$ to $w'_i$, while leave the remaining values unchanged. Let $W'$ be the modified value vector, and $w'^\downarrow_i$ be the $i$-th largest value in vector $W'$. A crucial observation is that there must be $w'^\downarrow_i \geq w^\downarrow_i$ for every $i \in [n]$. Then, we have
\[
\sum_{i = 1}^m \big(Z^\downarrow(\tilde S) \big)_i \cdot w'^\downarrow_i ~\geq~ \sum_{i = 1}^m \big(Z^\downarrow(\tilde S) \big)_i \cdot w^\downarrow_i, 
\]
i.e., the objective achieved by vector $W'$ is at least the objective achieved by vector $W$. Therefore, the allocation phase of \Cref{alg:simple-general} is monotone, which implies the algorithm can be converted into a truthful auction algorithm via Myerson's Lemma.

\subsection{Constant Approximation for Stochastic Advertiser Values}

Now, we show that the approximation ratio can be further improved to a constant, assuming each $w_i$ is drawn from distribution $\D_i$. To be specific, we prove the following:

\begin{Theorem}
\label{thm:simple-stochastic}
    For Nearest-Neighbor model, assuming each advertiser’s valuation factorizes as $w_i \cdot u_j$, where $w_i$ is a private value of advertiser $i$ independently drawn from a publicly known distribution $\D_i$, and $u_j$ is a public weight associated with slot $j$.  Then, there exists an allocation algorithm running in $\OTild(n + m)$ time that achieves a $504$-approximation. Furthermore, the algorithm can be converted to a truthful mechanism.
\end{Theorem}

Now, we prove \Cref{thm:simple-stochastic}. Similar to \Cref{sec:simple-general}, we assume without loss of generality that $m \leq n$, as we can arbitrarily add advertisers with zero value into the instance. 

The main idea of the algorithm is that the slots can be pre-selected by resolving the instance with slot values $\{u_j\}$ and advertiser values $\{\gamma^{(k)}\}$, where we define
\begin{align}
    \gamma^{(k)} ~:=~ \mathop{\E}\limits_{X_i \sim \D_i}[X^\downarrow_k] \label{eq:muk-def}
\end{align}
to be the expectation of the $k$-th largest value $X^\downarrow_k$ of samples $X_1, \cdots, X_n$. Formally, we present our algorithm in \Cref{alg:simple-stochastic}.

\begin{algorithm}[tbh]
\caption{\textsc{Constant Approximation for Factorized Nearest-Neighbor Model}}
\label{alg:simple-stochastic}
\KwIn{Number of slots $m$ and advertisers $n$, distributions $\D_1, \cdots, \D_n$, distance function $d$}
\tcc{Slot Selection Phase}
Input slot weights $u_1, u_2, \cdots, u_m$, such that WLOG $u_1 \geq u_2 \geq \cdots \geq u_m$. \\
For $k \in [n]$, calculate $\gamma^{(k)}$ defined in \eqref{eq:muk-def}. \\
Solve the Nearest-Neighbor instance via the algorithm from \Cref{thm:nn-alloc-constant} with values $v_{i, j} = u_j \cdot \gamma^{(i)}$ and distance function $d(\cdot, \cdot)$. Let $\widehat M$ be the returned matching.\\
Let $\widetilde S = S(\widehat M)$ be the pre-selected slots. \\
\tcc{Allocation Phase}
Advertisers report values $w_1, w_2, \cdots, w_n$, such that $w_i \sim \D_i$.\\
For $i \in \min\{|\widetilde S|, n\}$, assign the advertiser with $i$-th largest value to the slot in $\widetilde S$ such that the value of $u_j \cdot \min_{j' \in \widetilde S \setminus \{j\}} d(j, j')$ is the $i$-largest. \\
Let $\widetilde M$ be the corresponding matching. \\
\KwOut{Matching $\widetilde M$.}
\end{algorithm}

\paragraph{Representing the objective as a norm.} Similar to \Cref{sec:simple-general}, we first represent the objective of $\widetilde M$ given by \Cref{alg:simple-stochastic} as a norm.

Given an instance of the factorized Nearest-Neighbor model where $w_1, w_2, \cdots, w_n$ are values of advertisers such that $w_i \sim \D_i$, and $u_1 \geq u_2 \geq \cdots \geq u_m \geq 0$ are weights of slots. For $n$-dimensional vector $W = (w_1, w_2, \cdots, w_n) \in \R^n_{\geq 0}$ that only contains non-negative values and subset $S \subseteq [m]$, we define function $g_S: \R^n_{\geq 0} \to \R_{\geq 0}$ as follows:
\begin{align*}
    g_S(W) = \sum_{i = 1}^m w^\downarrow_i \cdot \big(Z^\downarrow(S) \big)_i,
\end{align*}
where we recall notation $w^\downarrow_i$ represents the $i$-th largest value of vector $W$, function $Z(S)$ returns a vector, such that $\big(Z(S)\big)_j = \one[j \in S] \cdot u_j \cdot \min_{j' \in S \setminus \{j\}} d(j, j')$, and notation $\big(Z^\downarrow(S) \big)_i$ represents the $i$-th largest value of vector $Z(S)$. Then, the value of $g_S(W)$ represents the objective of the optimal matching for the factorized Nearest-Neighbor instance, assuming vector $W$ contains advertisers' values, and subset $S$ is the selected subset of slots, i.e., let $M^*$ be the optimal matching of the factorized Nearest-Neighbor instance, we can simplify the value of $\sw(M^*)$ as
    $\sw(M^*) ~=~ \max_{S \subseteq [m]} g_S(W).$
We further define the above expression as a function of $W$. That is, for vector $W \in \R^n_{\geq 0}$, we define
\begin{align*}
    h(W) ~:=~ \max_{S \subseteq [m]} g_S(W).
\end{align*}

Then, the following claim suggests that both $g_S(W)$ and $h(W)$ are norm functions:

\begin{Claim}
\label{clm:gh-norm}
    Function $g_S(W)$ with any $S \subseteq [m]$ and $h(W)$ are both monotone symmetric norms.
\end{Claim}

\begin{proof}
    Note that the definition of $g_S(W)$ is identical to the definition of function $f(Z; W^\downarrow)$ (with the only difference of value weights). Then, proving $g_S(W)$ is a monotone symmetric norm is identical to the proof of \Cref{clm:f-norm}; we omit the detailed proofs for simplicity. Furthermore, Since the pointwise maximum of a collection of monotone symmetric norms remains a monotone symmetric norm, it follows that function $h$ is still a monotone symmetric norm, as it is the maximum over $\{g_S\}$.
\end{proof}

\paragraph{Proof of the approximation ratio.} Now, we are ready to prove the approximation ratio for \Cref{alg:simple-stochastic}. Our proof relies on the following lemma:

\begin{Lemma}[Theorem 5.1 in \cite{IS-FOCS20}]
\label{lma:focs20}
    Let vector $W \in \R^n_{\geq 0}$ follow a product distribution on $\R^n_{\geq 0}$, and $f: \R^n \to \R_{\geq 0}$ be a monotone symmetric norm. Then $f(\E[W^\downarrow]) \leq \E[f(W)] \leq 28 \cdot f(\E[W^\downarrow])$.
\end{Lemma}

We apply \Cref{lma:focs20} with $W = (w_1, \cdots, w_n)$ being the vector of advertisers' values. Note that vector $W$ follows a product distribution, such that $w_i \sim \D_i$, and vector $\E[W^\downarrow]$ is exactly $(\gamma^{(1)}, \gamma^{(2)}, \cdots, \gamma^{(n)})$. For simplicity of the notation, we define vector $\Gamma = \E[W^\downarrow] = (\gamma^{(1)}, \cdots, \gamma^{(n)})$.

Recall that when advertiser's value vector $W$ is fixed, the objective $\sw(M^*)$ is defined as $h(W)$. By taking the expectation over randomness of $W$ and applying \Cref{lma:focs20}, we have
\begin{align*}
    \E[\sw(M^*)] ~=~ \E[h(W)] ~\leq~ 28 \cdot h(\Gamma),
\end{align*}
where the inequality uses the fact that function $h$ is a monotone symmetric norm. A key observation is that $h(\Gamma)$ is exactly the objective of the optimal matching of the instance we constructed in Line 3 of \Cref{alg:simple-stochastic}. Let function $\widehat{\sw}(M)$ be the objective of matching $M$ for this instance, and $\widehat M^*$ be the optimal matching of this instance. Then, we have
\begin{align*}
    h(\Gamma) ~=~ \widehat{\sw}(\widehat M^*) ~\leq~ 18 \cdot \widehat{\sw}(\widehat M) ~\leq~ 18 \cdot g_{\widetilde S}(\Gamma),
\end{align*}
where the first inequality follows from the fact that the algorithm in \Cref{thm:nn-alloc-constant} gives a $18$-approximation, and the second inequality follows from the fact that $\widetilde S = S(\widehat M)$, and when $\widetilde S$ is fixed, the optimal matching is nothing but matching values in $\Gamma$ and $Z(\widetilde S)$ in descending order. 

Finally, since function $g_{\widetilde S}$ is a monotone symmetric norm,
\begin{align*}
    g_{\widetilde S}(\Gamma) ~\leq~ \E[g_{\widetilde S}(W)].
\end{align*}
Merging the above three inequalities together gives
\[
\E[\sw(M^*)] ~\leq~ 28 \cdot h(\Gamma) ~\leq~ 28 \cdot 18 g_{\widetilde S}(\Gamma) ~\leq~ 504 \cdot \E[g_{\widetilde S}(W)],
\]
which proves the approximation ratio of \Cref{alg:simple-stochastic}, as $\E[g_{\widetilde S}(W)]$ exactly represents the performance of \Cref{alg:simple-stochastic}.

\paragraph{Truthfulness and efficiency.} It remains to show that the allocation phase of \Cref{alg:simple-stochastic} runs in $\OTild(n + m)$ time, and the objective achieved in allocation phase is monotone, which implies that \Cref{alg:simple-stochastic} can be converted into a truthful auction algorithm via Myerson's Lemma. We omit the proofs of these two arguments: as the allocation phase of \Cref{alg:simple-stochastic} is identical to the allocation phase of \Cref{alg:simple-general}, the proofs are are identical to the proofs of truthfulness and efficiency for \Cref{alg:simple-general}.

\section{PTAS Algorithm for Unweighted Advertisers}
\label{appx:ptas_2d}

In this section, we present a PTAS algorithm that achieves a $(1 + \epsilon)$-approximation for the factorized Nearest-Neighbor model with unweighted advertisers and 2D Euclidean metric, that is, each valuation $v_{i, j}$ equals to $u_j$, which only depends on the identity of a slot, and the metric space is the Euclidean distance in the 2D plane. To be specific, we prove the following:

\begin{Theorem}
\label{thm:ptas}
    For Nearest-Neighbor model with $n$ advertisers and $m$ slots, assuming each $v_{i, j} = u_j$ and the metric $d (\cdot, \cdot)$ is  given by Euclidean distances in the 2D plane, given $\epsilon \in (0, 0.5)$, there exists an algorithm that runs in $O(m^{O(\epsilon^{-2})} \cdot \poly(n, \epsilon^{-1}))$ time and achieves a $(1 + \epsilon)$-approximation against the optimal allocation.
\end{Theorem}

Our algorithm consists of two steps. First, we reduce the problem to a Weighted Disk Selection problem. This problem is similar to the Grouped Pseudo-Disk Selection problem defined in \Cref{def:gpds}, where the metric is the Euclidean distance on a 2D plane. The main difference is that, in the Weighted Disk Selection problem, the disks are not partitioned into multiple groups that restrict the selection to at most one disk per group. Instead, we only require that the total number of selected disks be bounded. Next, we solve the grouped pseudo-disk selection problem via a PTAS algorithm. Our main algorithm is a dynamic programming algorithm based on \cite{EJS-SICOMP05}, which we modify to accommodate the requirements of the Weighted Disk Selection problem, and the detailed proofs are included for completeness. 

\subsection{Reduction to the  Weighted Disk Selection Problem}

To begin, we give the definition of the Weighted Disk Selection problem and the reduction:

\begin{Definition}[Weighted Disk Selection Problem]
\label{def:wds}
Define a disk $D = (c, r, w)$ on a 2D plane to be the set of points with distance \emph{strictly smaller} than $r$ from $c$, and is associated with weight $w$. Given $N$ disks $D_1 = (c_1, r_1, w_1), \cdots, D_N = (c_N, r_N, w_N)$ and parameter $M$, the \emph{Weighted Disk Selection Problem} asks to find a selection $T \subseteq [N]$ that chooses at most $Q$ disks from the collection with the objective of maximizing
\[
\val(T) ~:=~ \sum_{t \in T} w_t,
\]
such that the center of each selected disk $c_t: t \in T$ is not covered by another selected disk
\end{Definition}

\begin{Lemma}
\label{lma:reduction-wds}
For Nearest-Neighbor model with $n$ advertisers and $m$ slots, assuming each $v_{i, j} = u_j$ and the metric $d (\cdot, \cdot)$ is  given by Euclidean distances in the 2D plane. Given $\epsilon \in (0, 0.5)$, the allocation of this Nearest-Neighbor model admits a polynomial-time reduction to a Weighted Disk Selection instance with $N \leq m^2$ disks, such that the Weighted Disk Selection instance is a $(1 + \epsilon)$-approximation of the Nearest-Neighbor model. To be specific, the reduced instance satisfied the following:
\begin{enumerate}
    \item The maximum radius of the constructed disk is at most $1$, and the minimum radius of the constructed disk is at least $\epsilon/m$.
    \item Let $M^*$ be the optimal matching of the Nearest-Neighbor instance and $T^* \subseteq [N]$ be the optimal selection of the Weighted Disk Selection instance. Then, $(1 - \epsilon) \cdot \sw(M^*) \leq \val(T^*)$.
    \item Given a feasible selection $T \subseteq [N]$ of the Weighted Disk Selection instance, we can find a feasible matching $M$ for the Nearest-Neighbor instance in polynomial time, such that  $\sw(M) \geq \val(T)$.
\end{enumerate}
\end{Lemma}

\begin{proof}

 Given the Nearest-Neighbor instance, let $n$ be the number of advertisers, $m$ be the number of slots, $c_j \in \R^2$ be the location of slot $j$ on the 2D plane, and weights of slots $u_1, \cdots, u_m$, 
we build an instance of the Weighted Disk Selection problem as follows:  For each slot pair $(j, j')$ that satisfies $j, j' \in [m]$ and $j \neq j'$, we create a disk $D_{j, j'} = \big(c = c_j, r = d(j, j'), w = u_j \cdot d(j, j')\big)$, i.e., the center of the disk is the location of slot $j$, the radius of the disk is the distance between $j$ and $j'$, and the weight of the disk is the weight of the slot times the distance between $j$ and $j'$. Then, we add $D_{j, j'}$ into the Weighted Disk Selection instance if it satisfies $d(j, j') \geq \epsilon/m$. Note that this guarantees the correctness of the first statement in \Cref{lma:reduction-wds}, as the condition $r \leq 1$ is already satisfied by the assumption that the diameter of the metric space is $1$. Finally, we set parameter $Q = n$, i.e., the Weighted Disk Selection instance can select at most $n$ disks.

\paragraph{Proof of the second statement.} To prove the second statement, let $T^*$ be the optimal selection of the constructed Weighted Disk Selection instance, and let $M^*$ be the optimal matching of the   We aim to show
\begin{align}
    ( 1- \epsilon) \cdot \sw(M^*) ~=~ ( 1- \epsilon) \cdot \sum_{(i, j) \in M^*} u_j \cdot \min_{j' \in S(M^*) \setminus \{j\}} d(j, j') ~\leq~ \sum_{(j, j'): D_{j,j'} \in T^*}  u_j \cdot d(j, j') . \label{eq:ptas-disk-opt-to-opt}
\end{align}

To show \eqref{eq:ptas-disk-opt-to-opt}, we first construct a feasible matching $M'$ that removes the slots in $M^*$ that are too close to another matched slot. We define
\[
M' = \{(i, j): (i, j) \in M^* \land \min_{j' \in S(M^*) \setminus \{j\}} d(j, j') \geq \epsilon/m\}.
\]
Note that
\begin{align*}
    \sw(M') ~&\geq~ \sw(M^*) - \sum_{j \in S(M^*) \setminus S(M')} u_j \cdot \epsilon/m \\
    ~&\geq~ \sw(M^*) -  m \cdot \epsilon/m \cdot \max_{j \in [m]} u_j \\
    ~&\geq~ (1 - \epsilon) \cdot \sw(M^*),
\end{align*}
where the first inequality uses the fact that each slot $j$ removed from $M^*$ is at most $\epsilon/m$ away from another slot, the second inequality upper-bounds the value of each $u_j$ by the highest weight slot, and the last inequality uses the fact that only matching one advertiser to the highest weight slot forms a feasible solution with objective $\max_{j \in [m]} u_j$, and the value of $\sw(M^*)$ could only be higher. 

As we construct a feasible matching $M'$, such that $\sw(M') \geq (1 - \epsilon) \cdot \sw(M^*)$, to prove \eqref{eq:ptas-disk-opt-to-opt}, it's sufficient to show that $\val(T^*) \geq \sw(M')$. Define
\[
T' ~:=~ \left \{D_{j, j'}: \exists i \text{ s.t. }(i, j) \in M' \land j' = \arg \min_{j' \in S(M') \setminus \{j\}} d(j, j')\right \}.
\]
to be a selection of pseudo disks in the Weighted Disk Selection problem. Note that $T'$ is a feasible selection, because
\begin{itemize}
    \item Disk $D_{j, j'}$ should exist, as the value of $\min_{j' \in S(M') \setminus \{j\}} d(j, j')$ is at least $\epsilon/m$.
    \item $T'$ selects at most $n$ disks, as matching $M'$ contains at most $n$ edges.
    \item If a center $c_j$ of the disk $D_{j, j'} \in T'$ is covered by another disk $D_{j_1, j'_1}$, by the definition of disks, we have $d(j_1, j) < d(j_1, j'_1)$. Since $D_{j, j'} \in T'$, there must be $j \in S(M^*)$. Then, the inequality $d(j_1, j) < d(j_1, j'_1)$ is in contrast to the assumption that $j'_1$ minimizes the distance between $j_1$ and vertices in $S(M^*) \setminus \{j_1\}$.
\end{itemize}
Since $T'$ is feasible, By arguing that $\val(T') \leq \val(T^*)$, we have
\begin{align*}
     \sum_{(j, j'): D_{i,j,j'} \in T^*}  u_j \cdot d(j, j') ~\geq~ \sum_{(j, j'): D_{j,j'} \in T'}  u_j \cdot d(j, j') ~=~ \sum_{(i, j) \in M'} u_j \cdot \min_{j' \in S(M') \setminus \{j\}} d(j, j'),
\end{align*}
where the equality follows from the definition of $T'$. Combining the above inequality with the fact that $\sw(M') \geq (1 - \epsilon) \cdot \sw(M^*)$ proves \eqref{eq:ptas-disk-opt-to-opt}.

\paragraph{Proof of the third statement.} To prove the third statement, we aim to show for any feasible selection $T$, there exists a feasible matching $M$ such that
\begin{align}
        \sw(M) ~\geq~  \sum_{(j, j_1): D_{j,j_1} \in T} u_j \cdot d(j, j_1). \label{eq:ptas-disk-any-to-any}
\end{align}

To prove \eqref{eq:ptas-disk-any-to-any}, we construct matching $M$ via the following process: we match an arbitrarily advertiser to a slot $j$ that satisfies $\exists j_1 \neq j: D_{j, j_1} \in T$. The process stops until we match $|T|$ advertisers.

Since the Weighted Disk Selection instance requires $|T| \leq n$, matching $M$ is feasible, and its objective is 
\[
\sw(M) ~=~ \sum_{j: \exists j_1 \neq j \text{ s.t. } D_{j, j_1} \in T} u_j \cdot \min_{j' \in S(M) \setminus \{j\}} d(j, j').
\]
Note that for a fixed $j$, if there exists $j_1$ that satisfies $D_{j, j_1} \in T$, such $j_1$ must be unique, as the center $c_j$ of $D_{j, j_1}$ can be covered only once. Then, to show $\sw(M) \geq \val(T)$, we only need to show
\[
u_j \cdot d(j, j_1) ~\leq~ u_j \cdot \min_{j' \in S(M) \setminus \{j\}} d(j, j').
\]
We prove via contradiction: when the above inequality is not satisfied, there must exist $j, j_1, j', j'_1$ that satisfies the following:
\begin{itemize}
    \item $D_{j, j_1}, D_{j', j'_1} \in T$
    \item $j' = \arg \min_{j' \in S(M) \setminus \{j\}} d(j, j')$
    \item $d(j, j_1) > d(j, j')$.
\end{itemize}
This leads to a contradiction: as $d(j, j_1) > d(j, j')$, the center $c_{j'}$ of $D_{j', j'_1}$ is also covered by the disk $D_{j, j_1}$, which is impossible when $T$ is feasible. Then, the inequality $u_j \cdot d(j, j_1) ~\leq~ u_j \cdot \min_{j' \in S(M) \setminus \{j\}} d(j, j')$ holds, and therefore $\sw(M) \geq \val(T)$.
\end{proof}

\subsection{PTAS Algorithm via Dynamic Programming}

In this subsection, we provide the PTAS algorithm for the Weighted Disk Selection problem. We will show the following:

\begin{Theorem}
\label{thm:ptas-wds}
    For the Weighted Disk Selection problem with $N$ disks, given $\epsilon \in (0, 0.5)$, there exists an algorithm that runs in $O(\log (d_{max}/d_{min}) \cdot N^{O(\epsilon^{-2})})$ time, where $d_{max}$ and $d_{min}$ are the maximum and the minimum diameters of disks, respectively, and achieves at least $(1 - \epsilon)$ times the value of the optimal selection.
\end{Theorem}

To prove \Cref{thm:ptas-wds}, we first give the PTAS algorithm. The algorithm is similar the PTAS algorithm in \cite{EJS-SICOMP05} for the Weighted Independent Disk Selection problem, which is similar to our Weighted Disk Selection problem with the following two differences:
\begin{itemize}
    \item It requires the selected disks to be disjoint.
    \item The number of selected disks is not restricted.
\end{itemize}

Our proof follows the outline of the algorithm in \cite{EJS-SICOMP05}, and we omit the proofs of some lemmas when it is identical to the one in \cite{EJS-SICOMP05}.

\paragraph{Pre-processing.} Given $\epsilon$, we set $k = 1/\epsilon$ and assume without loss of generality that $\alpha$ is an integer. We also re-scale the 2D plane, so that the largest disk has diameter $1$. Let $d_{min}$ be the diameter of the smallest disk, we define $l := \log_{k+1} \lfloor 1/d_{min}\rfloor$ to be the number of levels we wish to partition. Then, we partition the set of disks into $l+1$ levels, such that for $j \in \{0, 1, \cdots, l\}$, level $j$ includes all disks $D_i$ with diameter $2r_i \in \big((k+1)^{-j-1}, (k+1)^{-j} \big]$. We also require the smallest disk with diameter $d_{min}$ being in level $l$.

For level $j$, we cut the plane into multiple grids with lines that are $k \cdot (k+1)^{-j}$ away from each other, both vertically and horizontally. For each vertical line, we require its x-coordinate equals to a multiple of $(k+1)^{-j}$, and that this multiple module $k$ equals a fixed value $\alpha \in \{0, 1, 2\cdots, k-1\}$. Similarly, for each horizontal line, we require its y-coordinate equals to a multiple of $(k+1)^{-j}$, and that this multiple module $k$ equals a fixed value $\beta \in \{0, 1, 2\cdots, k-1\}$. The above process partitions level $j$ into multiple grids. For a pair of fixed $(\alpha, \beta)$,  we rule out all disks that are intersecting the lines used for gridding from consideration. Then, the following lemma suggests that this gridding step does not harm our optimal solution a lot:

\begin{Lemma}
\label{lma:gridding}
    Let $T^*$ be the optimal selection of the Weighted Disk Selection instance, and let $T^*_{\alpha, \beta}$ be the optimal selection of the Weighted Disk Selection instance after removing the disks intersecting the lines used for gridding parameterized by $\alpha$ and $\beta$. Then, there exists a pair $(\alpha, \beta)$ that satisfies $\val(T^*_{\alpha, \beta}) \geq (1 - \epsilon)^2 \cdot \val(T^*)$.
\end{Lemma}

We omit the proof of \Cref{lma:gridding}, as it is identical to the proof of Lemma 2.1 in \cite{EJS-SICOMP05}. With \Cref{lma:gridding}, we can enumerate all pairs $(\alpha, \beta)$, and solve the gridded instance. This is efficient since there are 
only $O(k^2) = O(\epsilon^{-2})$ pairs of parameters $(\alpha, \beta)$. Furthermore, after fixing $(\alpha, \beta)$, for a grid formed by the gridding, we call it a $j$-grid if it is constructed in level $j$. Then, we have the following:

\begin{Lemma}
    \label{lma:grid-structure}
    For any $j \in \{0, 1, \cdots, l - 1\}$, every $(j+1)$-grid is fully contained in a $j$-grid. Furthermore, every $j$-grid is the union of $(k+1)^2$ number of $(j+1)$-grids.
\end{Lemma}

We omit the proof of \Cref{lma:grid-structure}, as it is identical to the proofs of Lemma 2.2 and Corollary 2.3 in \cite{EJS-SICOMP05}. With \Cref{lma:grid-structure}, we may build a tree structure for all grids, such that each $j$-grid contains $(k+1)^2$ children, while each child is a $(j+1)$-grids. We also construct a root node for this tree structure, which includes all $0$-grids. This root node refers to the boundary of the whole instance. For simplicity, we assume the root node is at level $-1$, and therefore the root is a virtual $-1$-grid, such that it has at most $N$ children.

\paragraph{The dynamic programming algorithm.} Now, we present our dynamic programming algorithm on the tree structure. The dynamic programming table contains the following dimensions:

\begin{itemize}
    \item $j \in \{-1, 0, 1, \cdots, l\}$: the level we are in.
    \item A $j$-grid $g$ that contains at least one disk.
    \item A subset $S$ of disks with level \textbf{lower} than $j$, such that each disk in subset $S$ should intersect grid $g$. Furthermore, subset $S$ should be feasible, such that no center of disk in subset $S$ is covered by another disk in $S$.
    \item $q \in \{0, 1, \cdots, Q\}$: the number of disks with level \textbf{at least} $j$ that we want to select.
\end{itemize}

When $(j, g, S, q)$ is determined, the dynamic programming table $dp(j, g, S, q)$ records the maximum weight (and the corresponding selection) of the $q$ disks with level at least $j$, such that the union of $S$ and $q$ selected disks is feasible. Our final output is exactly the maximum weight of $(j = -1, g = \text{virtual root note}, S = \varnothing, q \in \{0, 1, \cdots, Q\})$.

We run the dynamic programming algorithm in order of non-decreasing level $j$. When reaching $(j, g, S, q)$, the dynamic programming algorithm is as follows:

\begin{itemize}
    \item Enumerate subset $S'$ that contains level-$j$ disks in $g$, such that $S \cup S'$ is still feasible.
    \item Use a dynamic-programming subroutine to determine the optimal selection of the remaining $q - |S| - |S'|$ disks. To be specific, let $C_g$ be the set of grid $g$'s children grid.  The first dimension of the dynamic-programming table of the subroutine is the identity of the current child $g' \in C_g$, and the second dimension is an integer between $0$ and $q - |S| - |S'|$, which indicates the number of disks we have selected till the children we are currently enumerating. With this subroutine, we may update
    \begin{align*}
        dp(j, g, S, q) = \max_{S': S \cup S' \text{ is feasible}} ~\max \left\{\sum_{g' \in C_g}dp\left(j+1, g', (S \cup S') \cap \{D_i: D_i \cap g' \neq \varnothing\}, q_{g'}\right)  \right\},
    \end{align*}
    where the values $\{q_{g'}\}$ should satisfy
    \[
     \sum_{g' \in C_g} q_{g'} = q - |S| - |S'|.
    \]
\end{itemize}

\paragraph{Feasibility and optimality of the algorithm.} To verify the feasibility of the dynamic programming algorithm, it only needs to be noted that whether the selection of disks with level at least $j$ contained in a $j$-grid $g$ would effect the feasibility of the selection only depends on those selected disks with level lower than $j$, such that the center of the selected disk is in $g$ (and therefore the selected disk intersects $g$). This is because for a selected disk with center outside of $g$, no disk in $g$ can possibly cover its center. On the other hand, to feasibly select disks in $g$ with level at least $j$, it's sufficient when only the selected disks that intersect grid $g$ are revealed. Then, the dynamic programming provides sufficient information, and therefore the algorithm provides a feasible solution.

Furthermore, since $dp(j, g, S, q)$ records the maximum weight of the $q$ disks with level at least $j$ such that the union of $S$ and $q$ selected disks is feasible, the algorithm guarantees that the optimal solution $T^*_{\alpha, \beta}$ can be found by the algorithm. Therefore, the dynamic programming algorithm is also optimal.

\paragraph{Running time analysis.} Finally, we analyze the running time of the dynamic programming algorithm.

We first bound the size of the dynamic programming table. Dimension $j$ is at most $O(l) = O(\log (d_{max}/d_{min})$. Since each grid $g$ must contain at least one disk, there are at most $N$ grids to be enumerated. The dimension $q$ is bounded by $Q$, which is at most $N$. It only remains to bound the number of possible subsets $S$. We provide the following claim:

\begin{Claim}
\label{clm:intersect-bounded}
    Let $S$ be a subset of disks with level lower than $j$, such that $S$ is feasible, and each disk in $S$ intersects a $j$-grid $g$. Then, there must be $|S| \leq O(k^2)$.
\end{Claim}

\begin{proof}
    Recall that the side length of $j$-grid $g$ is $k \cdot (k+1)^{-j}$. We prove \Cref{clm:intersect-bounded} by considering disks in $S$ with diameter at most $12k \cdot (k+1)^{-j}$ and at least $12k \cdot (k+1)^{-j}$ separately, i.e., let
    \[
    S_1 = \{D_i \in S: 2\cdot r_i \leq 12k \cdot (k+1)^{-j}\} \quad \text{and} \quad S_2 = \{D_i \in S: 2\cdot r_i > 12k \cdot (k+1)^{-j}\},
    \]
    we bound $|S_1|$ and $|S_2|$ separately.
    
    \paragraph{Number of disks with small diameter.} We first consider the size of $|S_1|$. Consider a grid $\hat g$ (which is not a grid given by the gridding process) with side length $(k + 2\cdot 12k) \cdot (k+1)^{-j} = 25k \cdot (k+1)^{-j}$, such that $g$ is at the center of $\hat g$. Then, for a disk $D \in S_1$ that intersects $g$, it must be contained in this expanded grid $\hat g$. 
    
    For disks $D_i, D_j \in S_1$, note that if we halve the radius of both disks, the halved disk must be disjoint. This is because if the halved $D_i$ and $D_j$ intersect, it implies $d(i, j) < \frac{1}{2} \cdot (r_i + r_j)$, and therefore the larger one of $r_i$ and $r_j$ must be at least $d(i, j)$, which is in contrast to the assumption that the center of both $D_i$ and $D_j$ can't be covered by any other disk.  Then, since the total area of halved disks in $S_1$ is upper-bounded by the size of expanded grid $\hat g$, we have
    \begin{align*}
        25^2 k^2 \cdot (k+1)^{-2j} ~\geq~ \sum_{i: D_i \in S_1} \pi \cdot \left(\frac{1}{2} r_i\right)^2 ~\geq~ |S_1| \cdot \frac{1}{16} \cdot (k+1)^{-2j},
    \end{align*}
    where the last inequality uses the fact that each disk in $S_1$ has level lower than $j$, and therefore the diameter is at most $(k+1)^{-2j}$. Rearranging the above inequality gives $|S_1| \leq 10000k^2/\pi$.

    \paragraph{Number of disks with large diameter.} Next, we consider the size of $|S_2|$. Let $D_i, D_j \in S_2$. For simplicity of the calculation, we define $a = \frac{r_i}{k \cdot (k+1)^{-j}}$ and $b = \frac{r_j}{k \cdot (k+1)^{-j}}$ and assume without loss of generality that $a \geq b$. Then, since $2r_i, 2r_j > 12k \cdot (k+1)^{-j}$, we have $a, b > 6$. Since the center of $D_j$ can't be covered by disk $D_i$, there must be $d(c_i, c_j) \geq a$.

    Let $c_g$ be the center of grid $g$, and let $a' = \frac{d(c_i, c_g)}{k \cdot (k+1)^{-j}}, b' = \frac{d(c_j, c_g)}{k \cdot (k+1)^{-j}}$. Since $D_i$ intersects $g$, there must be $a' \in [a - 1, a + 1]$: the difference between $r_i$ and $d(c_g, c_i)$ is at most half of the diameter of grid $g$, and we further relax the coefficient $1/\sqrt{2}$ to $1$. Similarly, we have $b' \in [b + 1, b - 1]$.

    Let $\theta$ be the angle between vector $c_j - c_g$ and $c_i - c_g$. By the law of cosines, we have
    \begin{align*}
        \cos \theta = \frac{a'^2 + b'^2 - d^2(c_i, c_j)}{2 a' \cdot b'} ~&\leq~ \frac{a'^2 + b'^2 - a^2}{2 a' \cdot b'} \\
        ~&\leq~ \frac{a'^2 + b'^2 - (a' - 1)^2}{2 a' \cdot b'} \\
        ~&\leq~ \frac{1}{b'} + \frac{b'}{2a'} \\
        ~&\leq~ \frac{1}{b - 1} + \frac{b + 1}{2(b - 1)} ~\leq~ 0.9,
    \end{align*}
    where the last line uses the fact that $b' \in [b - 1, b + 1]$, $a' \geq a - 1 \geq b - 1$, and $b > 6$.
    Therefore, we have $\theta > \arccos(0.9) > \frac{\pi}{10}$, i.e., for any two disks in $S_2$, the angle between their centers and $c_g$ is at least $\pi/10$, and there can be at most $2\pi / (\pi/10) \leq 20$ disks in $S_2$.

    Combining two cases together, we have $|S| \leq 10000k^2 + 20 = O(k^2)$.
\end{proof}

With \Cref{clm:intersect-bounded}, the number of feasible subsets $S$ is at most $N^{O(k^2)} = N^{O(1/\epsilon^2)}$. Therefore, the size of the dynamic programming table is at most $O(l \cdot N^2 \cdot N^{O(1/\epsilon^2)}) = O(\log(d_{max}/d_{min}) \cdot N^{O(1/\epsilon^2)})$.

Next, we bound the running time of computing one dynamic programming table entry. The algorithm needs two steps: the first step is to enumerate a feasible subset $S'$, and the next step is to update the value of the entry via a dynamic programming subroutine. Since the subroutine needs to enumerate over all children of $g$, the number of selected disks in all enumerated children, and the the number of selected disks in current children, the running time is at most $O(\max\{N, (k+1)^2\} \cdot N^2) \leq O(N^3 \cdot \epsilon^{-2|})$. Then, it remains to compute the number of feasible subsets $S'$. We provide the following claim:

\begin{Claim}
\label{clm:intersect-bounded-sprime}
    Given $j, g, S, q$, let $Ss$ be a subset of disks with level $j$, such that $S \cup S'$ is feasible, and each disk in $S'$ is inside $g$. Then, there must be $|S'| \leq O(k^2)$.
\end{Claim}

\begin{proof}
    For disks $D_i, D_j \in S'$, note that if we halve the radius of both disks, the halved disks must be disjoint, as otherwise it implies either the center of $D_i$ is covered by $D_j$, or the center of $D_j$ is covered by $D_i$, which is in contrast to the assumption that $S \cup S'$ is feasible. Then, since the total area of halved disks in $S'$ is upper-bounded by the size of grid $g$, we have
    \[
    k^2 \cdot (k+1)^{-2j} ~\geq~ \sum_{i: D_i \in S'} \pi \cdot \left(\frac{1}{2} r_i \right)^2 ~\geq~ |S'| \cdot \frac{1}{16} \cdot \frac{1}{k^2} \cdot (k+1)^{-2j}.
    \]
    Rearranging the above inequality gives $|S'| \leq 16k^2/\pi$.
\end{proof}

\Cref{clm:intersect-bounded-sprime} implies that the number of feasible $S'$ is at most $N^{O(k^2)} = N^{O(1/\epsilon^2)}$. Therefore, each dynamic programming table entry can be computed in $N^3 \cdot \epsilon^{-2} \cdot N^{O(1/\epsilon^2)} = N^{O(1/\epsilon^2)}$ time.

Multiplying the size of the dynamic programming table and the time for computing one entry together, we have the total running time to be $O(\log (d_{max}/d_{min}) \cdot N^{O(1/\epsilon^2)})$, which finishes the proof of \Cref{thm:ptas-wds}.

\subsection{Combining Everything Together}

To end this section, we combine \Cref{lma:reduction-wds} and \Cref{thm:ptas-wds} together and prove \Cref{thm:ptas}.

\begin{proof}[Proof of \Cref{thm:ptas}]
    Given a Nearest-Neighbor instance that satisfies the assumptions in \Cref{thm:ptas}, we first apply \Cref{lma:reduction} to reduce the problem to a Weighted Disk Selection instance in $\poly(n, m)$ time. Furthermore, the Weighted Disk Selection instance should have $N = O(m^2)$ disks, and the ratio between the maximum diameter and the minimum diameter of the instance is at most $m/\epsilon$. Next, we run the algorithm in \Cref{thm:ptas-wds}, and get a feasible selection $T$. By \Cref{lma:gridding} and \Cref{thm:ptas-wds}, we have $\val(T) \geq (1 - \epsilon)^2 \cdot \val(T^*)$.  Finally, we apply \Cref{lma:reduction} to convert the disk selection back to a feasible matching $M$ of the Nearest-Neighbor instance in $\poly(n, m)$ time. Therefore, by \Cref{lma:reduction-wds} and \Cref{thm:ptas-wds}, the total running time is 
    \[
    O\big(\poly(n, m) + \log(\frac{d_{max}}{d_{min}}) \cdot N^{O(\epsilon^{-2})}\big) ~\leq~ O\big(\poly(n) \cdot\log(\frac{m}{\epsilon}) \cdot (m^2)^{O(\epsilon^{-2})}\big) ~\leq~O\big(\poly(n, \epsilon^{-1}) \cdot m^{O(\epsilon^{-2})} \big).
    \]

    It remains to verify the value of $\sw(M)$. We have
    \[
    \sw(M) ~\geq~ \val(T) ~\geq~ (1 - \epsilon)^2 \cdot \val(T^*) ~\geq~ (1 - \epsilon)^3 \cdot \sw(M^*),
    \]
    where the first and the third inequality follows \Cref{lma:reduction-wds}, and the second inequality follows \Cref{thm:ptas-wds} and \Cref{lma:gridding}. Rescaling the value of $\epsilon$ gives a $(1 +\epsilon)$-approximation algorithm that satisfies \Cref{thm:ptas}.
\end{proof}
\section{Hardness of Product-Distance Model}
\label{appx:product_distance_model}

In this section, we prove the hardness result for the Product-Distance model. To be specific, we prove the following:

\begin{Theorem}
\label{thm:pd-hard}
The \textsc{Max-Independent-Set} (MIS) problem with $m$ vertices can be reduced to the  Product-Distance model in $O(\poly(m))$ time. Therefore, for every $\epsilon > 0$, no polynomial time algorithm can be approximate the Product-Distance model within a factor of $m^{1 - \epsilon}$ unless $P = NP$.
\end{Theorem}

Instead of providing a direct reduction from the MIS problem to the Product-Distance model, we introduce a new Maximum Sum of Exponential Degree (MSED) problem on an unweighted graph, which serves as a bridge between the MIS problem and the Product-Distance model. The problem is defined as follows:

\begin{Definition}[Maximum Sum of Exponential Degree Problem]
\label{def:msed}
Given an unweighted and undirected graph $G = (V, E)$, the Maximum Sum of Exponential Degree (MSED) problem aims at finding 
\[
\rho^*(G) ~:=~ \max_{S \subseteq V} \rho(U), \quad \text{such that} \quad \rho(U) ~:=~ \sum_{i \in U} 0.5^{\deg_U(i)}
\]
is maximized, where $\deg_U(i)$ represents the number of vertices $j \in S$ such that edge $(i, j) \in E$.
\end{Definition}

To prove \Cref{thm:pd-hard}, we first reduce the MIS problem to the MSED problem, and then reduce the MSED problem to the Product-Distance model. Our first reduction is as follows:

\begin{Theorem}
\label{thm:mis-to-msed}
The \textsc{Max-Independent-Set} (MIS) problem on graph $G = (V, E)$ can be reduced to the  Maximum Sum of Exponential Degree (MSED) problem in $O(\poly(|V|, |E|)$ time. Therefore, for every $\epsilon > 0$, no polynomial time algorithm can  approximate the MSED problem within a factor of $|V|^{1 - \epsilon}$ unless $P = NP$.
\end{Theorem}

\begin{proof}
    Given an unweighted and undirected graph $G = (V, E)$ for the MIS problem, let $I^* \subseteq V$ be the maximum independent set of $G$.  Consider to solve the MSED problem on the same graph $G$. We first note that
    \[
    \rho(I^*) = \sum_{i \in I^*} 0.5^{\deg_{I^*}(i)} ~=~ |I^*|.
    \]
    Therefore, there must be $\rho^*(G) \geq \rho(I^*) = |I^*|$, i.e., for the same graph, the objective of the MSED problem is at least the objective of the MIS problem.

    We prove via contradiction to show that the MSED problem does not admit a good approximation algorithm. Suppose there exists an algorithm $\alg$ that achieves a $|V|^{1 - \epsilon}$-approximation for the MSED problem in $O(\poly(|V|, |E|)$ time for some $\epsilon > 0$. Given graph $G = (V, E)$, let $S \subseteq V$ be the output of $\alg$, i.e., we have
    \begin{align}
        \rho(U) = \sum_{i \in S} 0.5^{\deg_{S}(i)} ~\geq~ \frac{1}{|V|^{1 - \epsilon}} \cdot \rho^*(G) ~\geq~ \frac{1}{|V|^{1 - \epsilon}} \cdot |I^*| \label{eq:hardness-msed-to-mis}
    \end{align}
    Next, we refine the solution $U$ via the following process:
    \begin{itemize}
        \item Find $(i, j) \in E$ such that $i, j \in U$. Terminate the algorithm if such edge $(i, j)$ does not exist.
        \item Assume without loss of generality that $\deg_U(i) \geq \deg_U(j)$, update $U \gets S \setminus \{i\}$.
        \item Continue the process by going back to Step 1.
    \end{itemize}
    Let $\widetilde U \subseteq S$ be the subset we get after the above refinement process. A crucial observation is that the objective $\rho(\widetilde U)$ is at least the value of $\rho(U)$. This inequality holds because whenever a vertex $i$ is removed from $U$, the objective $\rho(U)$ is non-decreasing: After removing $i$, we lose $0.5^{\deg_U(i)}$; however, the value of $\deg_U(j)$ drops by one, so we gain 
    \[
    0.5^{\deg_U(j) - 1} - 0.5^{\deg_U(j)} ~=~ 0.5^{\deg_U(j)} ~\geq~ 0.5^{\deg_U(i)},
    \]
    where the last inequality follows the assumption that $\deg_U(i) \geq \deg_U(j)$. Furthermore, for other vertices $i' \in U$, removing $i$ from $U$ could only decrease the value of $\deg_U(i')$, and therefore the objective $\rho(U)$ is non-decreasing.

    Note that the final subset $\widetilde U$ after the refinement process is an independent set, as otherwise the refinement process should continue. Therefore, we have 
    \[
    \rho(\widetilde U)  = \sum_{i \in \widetilde U} 0.5^0 = |\widetilde U|.
    \]
    By combining with the inequality that $\rho(\widetilde U) \geq \rho(U)$ and \eqref{eq:hardness-msed-to-mis}, we have
    \[
    |\widetilde U| ~\geq~ \frac{1}{|V|^{1 - \epsilon}} \cdot |I^*|,
    \]
    i.e., algorithm $\alg$ together with the $\poly(|V|, |E|)$ time refinement process gives a polynomial time algorithm that $|V|^{1 - \epsilon}$-approximates the MIS problem, which is in contrast to the standard hardness of the MIS problem (e.g., see \cite{Zuckerman-ToC07}). Therefore, no polynomial time algorithm can approximate the MSED problem within a factor of $|V|^{1 - \epsilon}$.
\end{proof}

Our second reduction is as follows:

\begin{Theorem}
\label{thm:msed-to-pd}
The Maximum Sum of Exponential Degree (MSED) problem on graph $G = (V, E)$ with $|V| = m$ can be reduced to the  allocation problem of the Product-Distance model with $m$ slots in $O(\poly(m))$ time. Therefore, for every $\epsilon > 0$, no polynomial time algorithm can  approximate the allocation problem of the Product-Distance model within a factor of $m^{1 - \epsilon}$ unless $P = NP$.
\end{Theorem}

Combining \Cref{thm:mis-to-msed} and \Cref{thm:msed-to-pd} together proves \Cref{thm:pd-hard}. Now, we prove \Cref{thm:msed-to-pd} to end the section:

\begin{proof}[Proof of \Cref{thm:msed-to-pd}]
    Let $G = (V, E)$ be an undirected graph with $|V| = m$ vertices. We construct a Product-Distance instance with $n = m$ advertisers and $m$ slots, such that each slot maps to one vertex in $V$. For vertex $u \in V$, we define $j(u)$ to be the slot in the Product-Distance instance corresponding to vertex $u$. Similarly, for slot $j \in [m]$, we define $u(j)$ to be the vertex in graph $G$ corresponding to slot $j$.

    The values and the distance function of the instance are as follows:
    \begin{itemize}
        \item Valuation: we have $v_{i, j} = 1$ for every advertiser $i$ and slot $j$.
        \item Metric: for slots $j_1, j_2 \in [m]$ such that $j_1 \neq j_2$, we set $d(j_1, j_2) = 0.5$ if edge $(j_1, j_2) \in E$; otherwise (i.e., $(j_1, j_2) \notin E$), we set $d(j_1, j_2) = 1$. This setting satisfies the triangle inequality, as for any $j_1, j_2, j_3$, we have
        \[
        d(j_1, j_2) ~\leq~ 1 ~=~ 0.5 + 0.5 ~\leq~ d(j_1, j_3) + d(j_3, j_2).
        \]
    \end{itemize}

    We prove via contradiction to show that the Product-Distance model does not admit a good approximation algorithm, i.e., a good approximation algorithm for the Product-Distance model implies a good approximation algorithm for the MSED problem.

    Suppose there exists an algorithm $\alg$ that achieves an $m^{1 - \epsilon}$-approximation for the Product-Distance model in $O(\poly(m))$ time for some $\epsilon > 0$. Given the above constructed Product-Distance instance, let $M^*$ be the optimal matching of the instance, and $M$ be the matching given by the algorithm $\alg$. Then, we have
    \begin{align}
        \sw(M) ~=~ \sum_{j \in S(M)} 0.5^{\sum_{j' \in S(M) \setminus \{j\}} \one[d(j, j') = 0.5]} ~\geq~ \frac{1}{m^{1 - \epsilon}} \cdot \sw(M^*), \label{eq:hardness-swm-good}
    \end{align}
    where the first equality follows from the fact that all advertisers are with valuation $1$. Therefore, the value from slot $j$ is exactly the product of $d(j, j')$ for every $j' \in S(M) \setminus \{j\}$, which is exactly $0.5$ to the power of number of $j'$'s that satisfies $d(j, j') = 0.5$.

    Now, we convert matching $M$ back to a feasible solution of the MSED problem: Let $U = \{u(j): j \in S(M)\}$ be the subset of vertices we construct for the MSED problem. Our goal is to show that $\rho(U)$ is a good approximation for the MSED problem. We have
    \begin{align}
        \rho(U) ~=~ \sum_{u \in U} 0.5^{\deg_U(u)} ~&=~ \sum_{u \in U} 0.5^{\sum_{u' \in U \setminus \{u\}} \one[(u, u') \in E]} \notag \\
        ~&=~ \sum_{j \in S(M)} 0.5^{\sum_{j' \in S(M) \setminus \{j\}} \one[d(j, j') = 0.5]} \notag \\
        ~&=~ \sw(M) ~\geq~ \frac{1}{m^{1 - \epsilon}} \cdot \sw(M^*),\label{eq:hardness-rhou-good}
    \end{align}
    where the second line uses the fact that by definition of subset $U$, every slot $j \in S(M)$ maps to a vertex $u(j) \in U$; conversely,
    every vertex $u \in U$ maps to slot $j(u) \in S(M)$; conversely. Furthermore, by the construction of the Product-Distance instance, condition $(u, u') \in E$ is equivalent to the condition $d(j(u), j(u')) = 0.5$. 

    \eqref{eq:hardness-rhou-good} gives a lower bound for $\rho(U)$. Let $U^* \subseteq V$ be the optimal solution of the MSED problem, i.e., $\rho^*(G) = \rho(U^*)$. It remains to give an upper-bound for $\rho(U^*)$. Consider the following matching $\widetilde M$ for the Product-Distance instance: we first select a subset of slots so that $S(\widetilde M) = \{j(u): u \in U^*\}$. Then, for every $j \in S(\widetilde M)$, we match an arbitrary advertiser $i \in [n]$ to slot $j$. This leads to a solution with objective
    \begin{align*}
        \sw(\widetilde M) ~&=~ \sum_{j \in S(\widetilde M)} 0.5^{\sum_{j' \in S(\widetilde M) \setminus \{j\}} \one[d(j, j') = 0.5]} \\
        ~&=~ \sum_{u \in U^*} 0.5^{\sum_{u' \in U^* \setminus \{u\}} \one[(u, u') \in E]} \\
        ~&=~ \sum_{u \in U^*} 0.5^{\deg_{U^*}(u)} ~=~ \rho(U^*),
    \end{align*}
    where the first line follows the same argument as \eqref{eq:hardness-swm-good}, and the second line follows the same argument as \eqref{eq:hardness-rhou-good}. Combining the above inequality with the fact that $\sw(\widetilde M) \leq \sw(M^*)$, we have
    \[
    \rho(U) ~\geq~ \frac{1}{m^{1 - \epsilon}} \cdot \sw(M^*) ~\geq~ \frac{1}{m^{1 - \epsilon}} \cdot \sw(\widetilde M) ~=~\frac{1}{m^{1 - \epsilon}} \cdot \rho(U^*),
    \]
    i.e., algorithm $\alg$ gives a polynomial time algorithm that $|V|^{1-\epsilon}$-approximates the MSED problem, which is in contrast to \Cref{thm:mis-to-msed}. Therefore, no polynomial time algorithm can approximate the Product-Distance model within a factor of $m^{1 - \epsilon}$.
\end{proof}

\bibliographystyle{alpha}
\bibliography{ref}
\appendix

\section{Omitted Proofs in \Cref{sec:LP}}
\label{sec:LP_Missing}

\subsection{Proof of \Cref{lma:reduction}}

\begin{proof}[Proof of \Cref{lma:reduction}]
Let the Nearest-Neighbor instance be given by advertisers $[n]$, slots $[m]$, nonnegative valuations $v_{i,j}$, and a normalized metric $d:[m]\times[m]\to[0,1]$ (page diameter $1$).

We build an instance of the Grouped Pseudo-Disk Selection Problem on the same metric space $(\X,d)$ with $\X=[m]$.  For each advertiser $i\in[n]$ (this becomes group $T_i$) and each  pair $(j,j')$ with $j\in[m]$ and $j'\in[m]\setminus\{j\}$, we create a pseudo-disk $ D_{i,j,j'} = \big(c_j, r = d(j, j')/2, w = v_{i,j} \cdot d(j, j')\big)$, where $c_j \in \X$  represents the point of slot $j$ in the metric space. Then, we have $|T_i| = m(m-1)\le m^2$ for all $i$, and there are $K = n$ groups of pseudo-disks.

\noindent \textbf{Proof of the first statement.} Let $T^*$ be the optimal selection that satisfies \Cref{cons:I}. We aim to show
\begin{align}
    \sw(M^*) ~:=~& \sum_{(i, j) \in M^*} v_{i,j} \cdot \min_{j' \in S(M^*) \setminus \{j\}} d(j, j') \notag \\
    ~\leq~& \sum_{(i,j, j'): D_{i,j,j'} \in T^*}  v_{i,j} \cdot d(j, j') . \label{eq:disk-opt-to-opt}
\end{align}

To show \eqref{eq:disk-opt-to-opt}, we first define
\[
T' ~:=~ \left \{D_{i, j, j'}: (i, j) \in M^* \land j' = \arg \min_{j' \in S(M^*) \setminus \{j\}} d(j, j')\right \}.
\]
to be a selection of pseudo disks in the Grouped Pseudo-Disk Selection problem. Note that $T'$ is a feasible selection, because if vertex $x \in \X$ is covered by two disks $D_{i_1, j_1, j'_1}, D_{i_2, j_2, j'_2} \in T'$, by the definition of pseudo disks, we have
\begin{align*}
    d(j_1, j_2) ~&\leq~ d(j_1, v) + d(j_2, v) \\
    ~&<~ 0.5 \cdot \big(d(j_1, j'_1) + d(j_2, j'_2)\big) \\
    ~&\leq~ 0.5 \cdot \big(d(j_1, j_2) + d(j_2, j_1)\big) ~=~ d(j_1, j_2),
\end{align*}
where the first line follows from triangle inequality, the second line follows from the definition of $D_{i_1, j_1, j'_1}, D_{i_2, j_2, j'_2}$, and the last line uses the fact that $j'_1$ is the vertex nearest to $j_1$ in $S(M^*)$. Since $j_2 \in S(M^*)$, there must be $d(j_1, j'_1) \leq d(j_1, j_2)$. Similarly, we have $d(j_2, j'_2) \leq d(j_2, j_1)$. Then, the assumption that $v$ is covered by two different disks leads to a contradiction, and therefore $T'$ is feasible. By arguing that $\val(T') \leq \val(T^*)$,
\begin{align*}
     \sum_{(i,j, j'): D_{i,j,j'} \in T^*}  v_{i,j} \cdot d(j, j') ~&\geq~ \sum_{(i,j, j'): D_{i,j,j'} \in T'}  v_{i,j} \cdot d(j, j') \\
     ~&=~ \sum_{(i, j) \in M^*} v_{i,j} \cdot \min_{j' \in S(M^*) \setminus \{j\}} d(j, j'),
\end{align*}
where the equality follows from the definition of $T'$.

\noindent \textbf{Proof of the second statement.} To prove the second statement, we aim to show for any feasible selection $T$ that satisfies \Cref{cons:II}, there exists a feasible matching $M$ such that
\begin{align}
        \sw(M) ~\geq~ 0.5 \cdot \sum_{(i,j, j_1): D_{i,j,j_1} \in T} v_{i,j} \cdot d(j, j_1). \label{eq:disk-any-to-any}
\end{align}
To prove \eqref{eq:disk-any-to-any}, we define
\[
M ~:=~ \left\{(i, j): \text{there exists } j' \text{ such that }D_{i,j,j'} \in T^*\right\}.
\]

Note that $M$ is a feasible matching, as if an advertiser $i$ is matched to two slots $j_1, j_2$, it is in contrast to the constraint that $|T \cup T_i| = 1$ for the Grouped Pseudo-Disk Selection Problem. On the other hand, if a slot $j$ is matched to two advertisers $i_1$ and $i_2$,  $c_j$ must be covered by two different disks in $T$ with $v$ serving as the center, which is in contrast to the assumption that $c_j$ can be covered by at most once. Therefore, $M$ is feasible. 

For $(i, j) \in M$, let $j_1$ be the point in the metric space such that $D_{i, j, j_1} \in T$. Such $j_1$ must be unique, as the center $c_j$ can be covered by at most once. Let $j_2 = \arg \min_{j' \in S(M) \setminus \{j\}} d(j, j')$. Note that there must be $0.5\cdot d(j, j_1) \leq d(j, j_2)$: since $j_2$ is a selected slot in $M$, there must be a corresponding disk in $T$ that uses $c_{j_2}$ as the center. If $0.5 \cdot d(j, j_1) > d(j, j_2)$ is also satisfied, it implies that $c_{j_2}$ is also covered by the disk $D_{i, j, j_1} \in T$, which is in contrast to \Cref{cons:II}. Then, $0.5\cdot d(j, j_1) \leq d(j, j_2)$ must be satisfied, which implies
\begin{align*}
    0.5\cdot \sum_{(i,j, j_1): D_{i,j,j_1} \in T} v_{i,j} \cdot d(j, j_1) ~&\leq~  \sum_{(i, j) \in M} v_{i,j} \cdot \min_{j' \in S(M) \setminus \{j\}} d(j, j') \\
    ~&=:~ \sw(M). \qedhere
\end{align*}
\end{proof}

\subsection{Proof of \Cref{thm:nn-auc}}

\begin{proof}[Proof of \Cref{thm:nn-auc}]
    By Myerson's Lemma \cite{Myerson1981}, it's sufficient to show that the algorithm for \Cref{thm:nn-alloc-constant} is \emph{monotone}, that is, given a Nearest-Neighbor instance, if we increase the values $\{v_{i, j}\}$ to $\{v'_{i,j}\}$ so that $v_{i, j} \leq v'_{i, j}$, the (expected) welfare from the matching with valuations $\{v_{i, j}\}$ is at most the (expected) welfare from the matching with valuations $\{v'_{i,j}\}$.

    We first show that after reduction, the Grouped Pseudo-Disk Selection problem is monotone. This follows from two observations: firstly, the objective of \eqref{lp:relax} is monotone, because when the optimal solution of \eqref{lp:relax} is fixed, increasing the weight $w_i$ of a disk does not change the feasibility of the LP solution, but the objective of the LP increases. Secondly, by \eqref{eq:equal}, the output $\widetilde T$ of \Cref{alg:rounding} guarantees that $\E[\val(\widetilde T)]$ exactly equals to $\frac{1}{9}$ times the objective of \eqref{lp:relax}. Therefore, the Grouped Pseudo-Disk Selection problem is monotone.

    Note that when reducing a Nearest-Neighbor instance to a Grouped Pseudo-Disk Selection instance, the monotonicity is preserved, i.e., increasing the value of $v_{i, j}$ leads to a larger weight $w_i$ for the Grouped Pseudo-Disk Selection instance. Then, to prove \Cref{thm:nn-auc}, it only remains to show that the monotonicity is preserved when converting the solution $\widetilde T$ back to a feasible matching. To achieve this, given $\widetilde T$, we perform the modified conversion: for every $D_{i, j, j'} \in \widetilde T$, 
    \begin{itemize}
        \item Add $(i, j)$ into $\widetilde M$.
        \item Add an empty slot $\hat j$ that satisfies $d(j, \hat j) = d(j, j')/2$ and a virtual advertiser $\hat i$ that has $0$ value for every slot into the instance, and add $(\hat i, \hat j)$  into   $\widetilde M$.
    \end{itemize}
    To verify the feasibility of $\widetilde M$, the feasibility of $\{(i, j): D_{i, j, j'} \in \widetilde T\}$ is guaranteed by \Cref{lma:reduction}, and the feasibility of $\{(\hat i, \hat j)\}$ follows from the fact that both $\hat i$ and $\hat j$ are newly added to the instance. Furthermore, note that after adding $(\hat i, \hat j)$ to the matching, each $(i, j): D_{i, j, j'}$ contributes exactly $v_{i, j} \cdot d(j, j')/2$ to $\sw(\widetilde M)$, which is $0.5$ times the weight of $D_{i, j, j'}$, and each $(\hat i, \hat j)$ contributes $0$ to $\sw(\widetilde M)$. Then, we have $\sw(\widetilde M) = 0.5 \cdot \val (\widetilde T)$, and therefore the monotonicity is preserved.
\end{proof}

\clearpage

\appendix

\end{document}